\documentclass[pdflatex,sn-mathphys-ay]{sn-jnl}% Math and Physical Sciences Numbered Reference Style
\usepackage{amsthm, amsmath, amssymb}
\usepackage{comment}
\usepackage{setspace}
\usepackage{multirow}
\usepackage{enumitem}
\usepackage{xcolor}
\usepackage[utf8]{inputenc}
\usepackage{natbib}
\usepackage{longtable}

\usepackage{url} % not crucial - just used below for the URL 
\usepackage{comment} % for commenting out section
\usepackage{natbib}
\usepackage{graphicx} % Required for inserting images
\usepackage{bbold}
\usepackage{pifont}
\usepackage{amssymb}
\usepackage{mathbbol}
\usepackage{float}
\usepackage{algpseudocode}
\usepackage{algorithm}
\usepackage{float}
\usepackage{xcolor}
\usepackage{tabularray}
\usepackage{bm} % Math packages
\usepackage{setspace}
\usepackage{multirow}
\usepackage{lscape}
\usepackage{rotating}
\usepackage{adjustbox}
\usepackage{booktabs}
\usepackage{ragged2e}
\usepackage{tikz}
\usepackage{makecell}

\usepackage{xparse}
\usepackage{subcaption}
\usepackage{float}
\def\*#1{\mathbf{#1}}
\def\^#1{\mathbf{#1}}
\def\##1{\mathbb{#1}}
\def \mathone{\mathbb{1}}
\DeclareSymbolFontAlphabet{\amsmathbb}{AMSb}%

\newtheorem{definition}{Definition}
\newtheorem{theorem}{Theorem}
\newtheorem{lemma}{Lemma}

\definecolor{blue}{RGB}{0,0,255}
\definecolor{slightgreen}{RGB}{34,139,34}
\definecolor{skyblue}{RGB}{30,180,255}
\definecolor{plum}{rgb}{0.56, 0.27, 0.52}

\def\A{\bm{A}} \def\a{\bm{a}}
   
\def\c{\bm{c}} \def\C{\bm{C}} 
 
\def\E{\mathbb{E}}

\def\K{\bm{K}}

\def\n{\bm{n}} \def\N{\mathcal{N}} 
\def\O{\mathcal{O}}

\def\bmS{\bm{S}}

 \def\u{\bm{u}}
\def\W{\bm{W}}  
\def\X{\bm{X}}  
\def\Y{\bm{Y}}   
\def\Z{\bm{Z}}  

\def\Var{\mathbb{Var}}

\def\bigO{\mathcal{O}}
\newcommand\lowero{
  \mathchoice
    {{\scriptstyle\mathcal{O}}}% \displaystyle
    {{\scriptstyle\mathcal{O}}}% \textstyle
    {{\scriptscriptstyle\mathcal{O}}}% \scriptstyle
    {\scalebox{.7}{$\scriptscriptstyle\mathcal{O}$}}%\scriptscriptstyle
  }

\NewDocumentCommand{\evalat}{sO{\bigg}mm}{%
  \IfBooleanTF{#1}
   {\mleft. #3 \mright |_{#4}}
   {#3#2|_{#4}}%
}

\def\1{\bm{1}}  

\def\Var{{\rm{Var}}}

\def\bbeta{\bm{\beta}}

\def\btheta{\bm{\theta}}

\def\bTheta{\boldsymbol{\Theta}}

\def\dconverge{\overset{d}{\rightarrow}}

\newcommand{\argmax}{\mathop{\mathrm{argmax}}} 
\newcommand{\argmin}{\mathop{\mathrm{argmin}}} 

\usepackage{atbegshi}% http://ctan.org/pkg/atbegshi

\usepackage{xcite}
\usepackage{xr}

\usepackage{graphicx}%
\usepackage{multirow}%
\usepackage{amsmath,amssymb,amsfonts}%
\usepackage{amsthm}%
\usepackage{mathrsfs}%
\usepackage[title]{appendix}%
\usepackage{xcolor}%
\usepackage{textcomp}%
\usepackage{manyfoot}%
\usepackage{booktabs}%
\usepackage{algorithm}%
\usepackage{algorithmicx}%
\usepackage{algpseudocode}%
\usepackage{listings}%
\theoremstyle{thmstyleone}%

\begin{document}

\title[
Voronoi-Elitism Genetic Algorithm: A Generic Derivative-Free Routine With Theory and Implementation for Statistical Optimization
% A Voronoi-Elitism Genetic Algorithm for Derivative-Free Statistical Optimization
]{
Voronoi-Elitism Genetic Algorithm: A Generic Derivative-Free Routine With Theory and Implementation for Statistical Optimization
% A Voronoi-Elitism Genetic Algorithm for Derivative-Free Statistical Optimization
}


\author[1]{\fnm{Anthony Haitao} \sur{Zou}}\email{loggamma@gmail.com}

\author*[2]{\fnm{Yizhou Jake} \sur{Cai}}\email{yizhouc@email.sc.edu}

\author[2]{\fnm{Ting Fung} \sur{Ma}}\email{tingfung@mailbox.sc.edu}

\affil[1]{\orgname{River Bluff High School}, \orgaddress{\street{320 Corley Mill Road}, \city{Lexington}, \postcode{29072}, \state{South Carolina}, \country{USA}}}

\affil[2]{\orgdiv{Department of Statistics}, \orgname{University of South Carolina}, \orgaddress{\street{1523 Greene Street}, \city{Columbia}, \postcode{29208}, \state{South Carolina}, \country{USA}}}

%\affil[3]{\orgdiv{Department of Mathematics}, \orgname{Texas State University}, \orgaddress{\street{201 Pickard St}, \city{San Marcos}, \postcode{78666}, \state{Texas}, \country{USA}}}


\abstract{In this paper, we propose a generic optimization approach for challenging objective functions that finds applications in various statistical problems. We focus on objective functions with two parameter blocks of one amenable to analytic optimization, and another that is irregular or computationally expensive. To address this setting, we propose the Voronoi-Elitism Genetic Algorithm (VEGA), a derivative-free optimization method that embeds geometric information into genetic search. The proposed algorithm retains elite candidates and constructs Voronoi-based neighborhoods around them, whose crossover and self-adaptive mutation balance exploitation of promising solutions with exploration of under-covered regions. We study the high dimensional behavior of genetic search by analyzing distance concentration, and the effects of population size and shrinking mutation, which shows that the algorithm improves spatial coverage and yields sharper distance bounds under limited computational budgets. 
Simulation studies are conducted to compare VEGA with two genetic-type algorithms competitors in finite samples. A real data application on Stack Exchange activity data further illustrates its ability to identify stable structural changes, implying the algorithm is computationally flexible for high-dimensional, derivative-free optimization and applicable for various statistical problems.}




\maketitle

\section{Introduction}
Parameter estimation is often formulated as an optimization problem. This perspective underlies a large class of statistical inference and learning, where model fitting is typically driven by the optimization of a loss function. When the objective function is convex and differentiable, estimators can often be obtained by characterizing the optimizer directly or by applying classical iterative procedures such as Newton-Raphson and gradient descent \citep{Ypma1995NewtonRaphson,Cauchy1847Gradient}. In less regular problems, however, optimization becomes substantially more challenging.

When the objective function is not differentiable, methods that rely directly on derivatives or gradients are no longer applicable in their classical form. If convexity is preserved, one may instead resort to sub-gradient-based methods. Adaptive first-order methods, including AdaGrad, ADADELTA, and Adaptive Moment Estimation (ADAM), are widely used in modern data science applications and can be implemented with sub-gradients \citep{DuchiHazanSinger2011AdaGrad,Zeiler2012AdaDelta,KingmaBa2015Adam}. Empirically, these methods often perform well for common objective and loss functions, illustrating their practical effectiveness in large-scale optimization. From a theoretical perspective, however, objective functions may be nonconvex, non-smooth, or otherwise irregular. In such settings, convergence analysis often requires additional structural assumptions, such as the Kurdyka-Lojasiewicz property \citep{AttouchBolteSvaiter2013KL}, while some procedures may still be used heuristically even when few regularity assumptions are available.

For particularly difficult objective functions, grid search is often used as a last-resort optimization strategy. Suppose, after suitable rescaling, the target function $f(\X)$ is defined on a compact domain $\bmS = [0,1]^p$. Let $a_0 < a_1 < \cdots < a_k$, where $a_0=0$ and $a_k=1$, be a partition of $[0,1]$ such that $\min_i |a_i-a_{i-1}| \rightarrow 0$. Applying this partition along each coordinate yields a $p$-dimensional grid over $\bmS$. The point attaining the largest function value on this grid provides a discrete approximation to $\argmax_{\X \in \bmS} f(\X)$. Although grid search is conceptually straightforward and can be accurate when the grid is sufficiently fine, it is computationally expensive \citep{Bergstra2011}. Its complexity is of order $\bigO(k^p)$, which quickly becomes unaffordable in moderate or high dimensional problems and reflects the curse of dimensionality \citep{Bellman1957DynamicProgramming}. This motivates the development of alternative optimization strategies in problems with exploitable structure.

A further limitation of grid search is that it is naturally designed for continuous parameter spaces, whereas many applications involve discrete, categorical, or mixed-type variables. Such settings arise frequently in hyperparameter optimization, where the search space may contain continuous, discrete, categorical, and conditional components \citep{Bergstra2011}. In practice, procedures such as the Stepwise Categorical Progress algorithm \citep{MILE} may be used in such settings. Discrete optimization is itself a broad and active area of research, but it lies outside the scope of this paper.

From the perspectives of statistics and data science, we are often interested in maximizing an objective function over a parameter space and treating the maximizer as the estimator. The class of problems we study occupies an intermediate regime that the objective is neither fully regular nor highly irregular objective functions. As examples, \citep{MILE} raised a specific scenario as following.

Suppose the target function is $f(\btheta, \Z)$, where $\btheta \in \Theta \subset \mathbb{R}^{k_1}$ and $\Z \in \mathcal{Z}\subset \mathbb{R}^{k_2}$ are parameter vectors for some $k_1, k_2 \in \mathbb{N}^+$. We separate the parameters into $\btheta$ and $\Z$ because they have different mathematical properties. Conditional on $\Z$, $f(\btheta, \Z)$ is mathematically well-behaved and easy to be optimized. In contrast, given $\btheta$, optimization over $\Z$ is complicated and computational challenging. This pattern captures the intermediate regime above. Our goal is to maximize this function and to get the joint maximizers, i.e.,

\begin{equation}
    (\widehat{\btheta}, \widehat{\Z}) = \argmax_{(\btheta, \Z) \in \Theta \times \mathcal{Z}} f(\btheta, \Z).
    \label{full}
\end{equation}

Notice that $f(\btheta,\Z)$ may be analytically complicated because of its dependence on $\Z$. In particular, it could ask for joint differentiability. Hence, many standard optimization methods, including Newton-Raphson and Block Ascending Gradient, would fail. In line with statistical applications, we therefore assume the target function is smooth on $\btheta$, while remaining $\Z$ as parameters that are numerically challenging to optimize.

Formally, we could descrine the intermediate structure as follows. Given any $\Z_0 \in \mathcal{Z}$, it is straightforward to derive $\widehat{\btheta}$, where
\begin{equation}
    \widehat{\btheta} = \argmax_{\btheta \in \Theta} f(\btheta, \Z_0),
    \label{easy}
\end{equation}
but on the reverse, it is still demanding to derive $\widehat{\Z}$ given $\btheta_0 \in \Theta$, where
\begin{equation}
    \widehat{\Z} = \argmax_{\Z \in \mathcal{Z}} f(\btheta_0, \Z).
    \label{difficult}
\end{equation}

An iterative algorithm can therefore be constructed by alternating between the optimization steps in \eqref{easy} and \eqref{difficult}, yielding a joint local estimator for \eqref{full}. Furthermore, under joint concavity of the objective function and suitable regularity conditions on the parameter space, a local maximizer is also a global maximizer. Although the optimization in \eqref{difficult} may remain computationally expensive, the alternating strategy is closely related to the principle of block coordinate ascent. Equivalently, it can be viewed as block coordinate descent applied to $-f$ \citep{Tseng2001BCD}.

\section{Background and Approaches}

\subsection{Optimization in Statistical Methods}

Optimization has been central to statistical methodology, not only as a
computational step but also as the fundamental definition of various estimators. Many estimators such as Maximum likelihood estimation (MLE), penalization estimator \citep{tibshirani1996regression}, variational inference \citep{VBReview}, and $M$-estimation \citep{huber1964robust} can all be written as optimization problems. In these settings, the statistical properties of an estimator are tied to the geometry of the objective function, whose identification is expressed through the population optimizer, while asymptotic behaviors depend on curvature and local behavior near the optimum
\citep{hansen1982large}.

The computational difficulty, however, varies substantially across statistical
problems. Smooth and convex objectives can often be handled by classical
gradient, Newton-type or coordinate-wise methods, whereas non-smooth, discrete and non-convex objectives require different treatment. Important examples include regression quantiles \citep{koenker1978regression}, maximum score estimation \citep{manski1975maximumScore}, and structural-change or segmented-regression models with unknown breakpoints \citep{bai1998estimating,Tseng2001BCD}.

The setting considered in this paper lies in this intermediate regime. Conditional
on the irregular block $\Z$, the regular block $\theta$ can often be optimized by
analytic or standard statistical procedures. In contrast, the optimization over $\Z$ is derivative-free and computationally expensive. This separation motivates the GA framework developed as following. The genetic algorithm searches the irregular $\mathcal{Z}$, while traditional estimation is retained for $\theta$. Specifically, the reliability of the framework depends heavily on the efficient of the search on $\mathcal{Z}$ and we propose a mechanism to improve the high dimensional performance under a limited computational budget. See \cite{kolda2003optimization,rios2013derivative} for examples.

\subsection{Genetic Algorithm}
The main computational difficulty in our framework is driven by the $\Z$ block, for which standard derivative-based approaches are generally unsuitable. Accordingly, we focus on the derivative-free methods employment.

Genetic algorithm (GA) is a class of derivative-free numerical methods designed for optimizing irregular or pathological objective functions. GA are efficient random search procedures that have been studied and modified over several decades. Figure \ref{fig:GASTEPS} illustrate the structure of a standard GA. In our setting, the GA is used to optimize over $\Z$, whereas $\btheta$ is obtained analytically under the assumed structure of the objective function.

GA mimics the evolution mechanism in nature, particularly the principle of survival of the fittest. The algorithm maintains a population of candidate solutions for $\widehat{\Z}$, which could be viewed as a gene pool of a species. Through repeated selection and reproduction, the population evolves toward regions of the parameter space associated with higher objective values. At the end of the procedure, the best performing individual is taken as the solution. The main components of the algorithm are encoding, fitness evaluation, crossover, and mutation.

\begin{itemize}
    \item \textbf{Encoding} maps each candidate solution to a chromosome representation, typically a vector or another equivalent structured object. Depending on the application, the chromosome may consist of binary, integer-valued, or real-valued components, which represent the underlying parameter values after decoding. Encoding could be viewed as a mapping $f:\mathcal{Z}\to\mathcal{A}$, where $\mathcal{A}$ denotes the chromosome space.

    \item \textbf{Fitness Evaluation} quantifies the quality of each candidate chromosome, after which the population is ranked according to fitness. This ranking determines the extent to which a candidate contributes genetic material to the next generation. Such contribution could be defined through different mechanisms, including survival probabilities or expected offspring counts. In unconstrained problems, the fitness function is typically the objective function itself, while in constrained problems, a penalized version is often used instead. Common examples include cross-entropy loss and negative log-likelihood function.

    \item \textbf{Crossover} produces new chromosomes from parents. Chromosomes are paired according to a specified selection rule, typically preferring individuals with higher fitness. Each pair serves as parents, and crossover combines chromosomes from the parental chromosomes to produce offspring for the next generation. These offspring form new candidate solutions for $\widehat{\Z}$.

    \item \textbf{Mutation} is the final step in each generation and used to introduce additional diversity into the population. The mutation mechanism depends on the representation of the chromosome. For example, real-valued components may be perturbed continuously, whereas binary components may be flipped with a specified probability.
\end{itemize}

For a comprehensive treatment of GA, see \cite{Eiben2015}.

\begin{figure}[htbp]
    \centering
       \begin{subfigure}[b]{0.48\linewidth}
        \centering
        \includegraphics[width=0.65\linewidth]{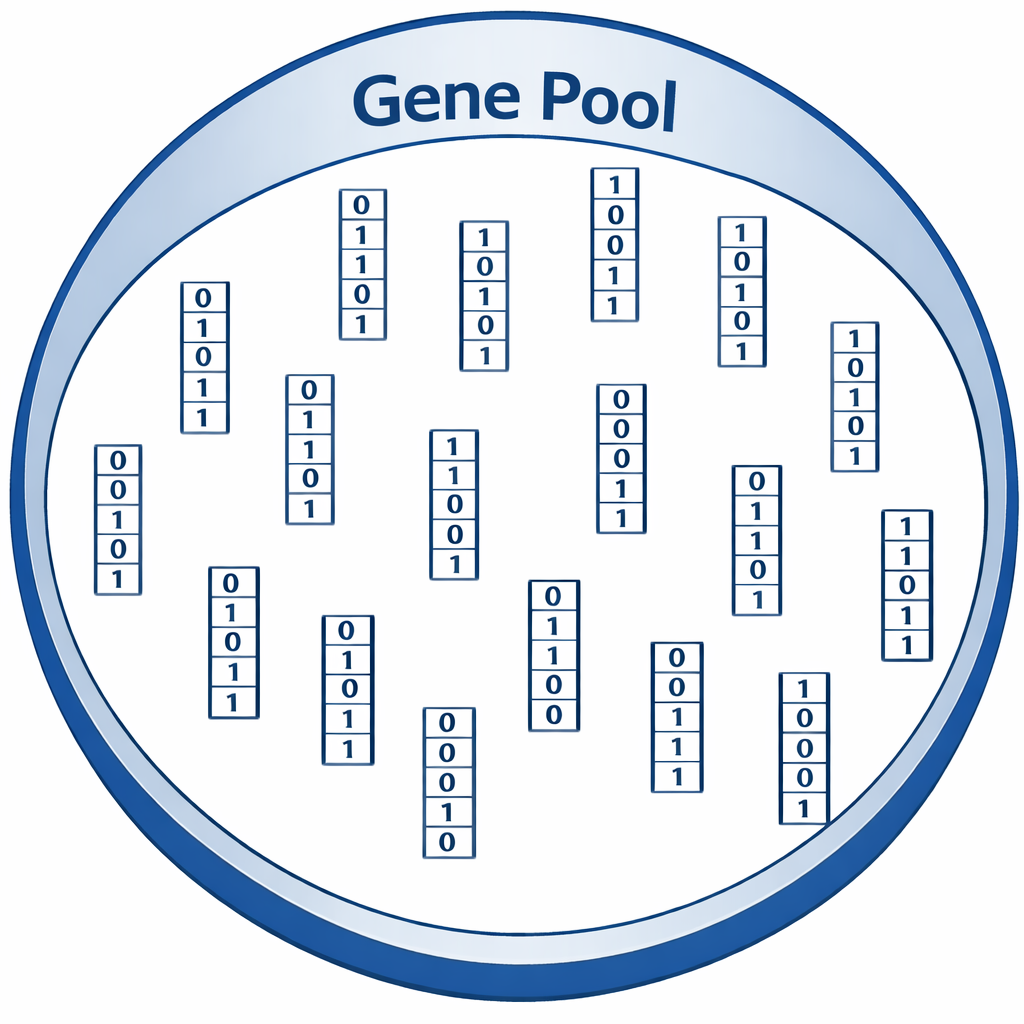}
        \caption{Encoding}
        \label{fig:sub1}
    \end{subfigure}
    \hfill
    \begin{subfigure}[b]{0.48\linewidth}
        \centering
        \includegraphics[width=1.10\linewidth]{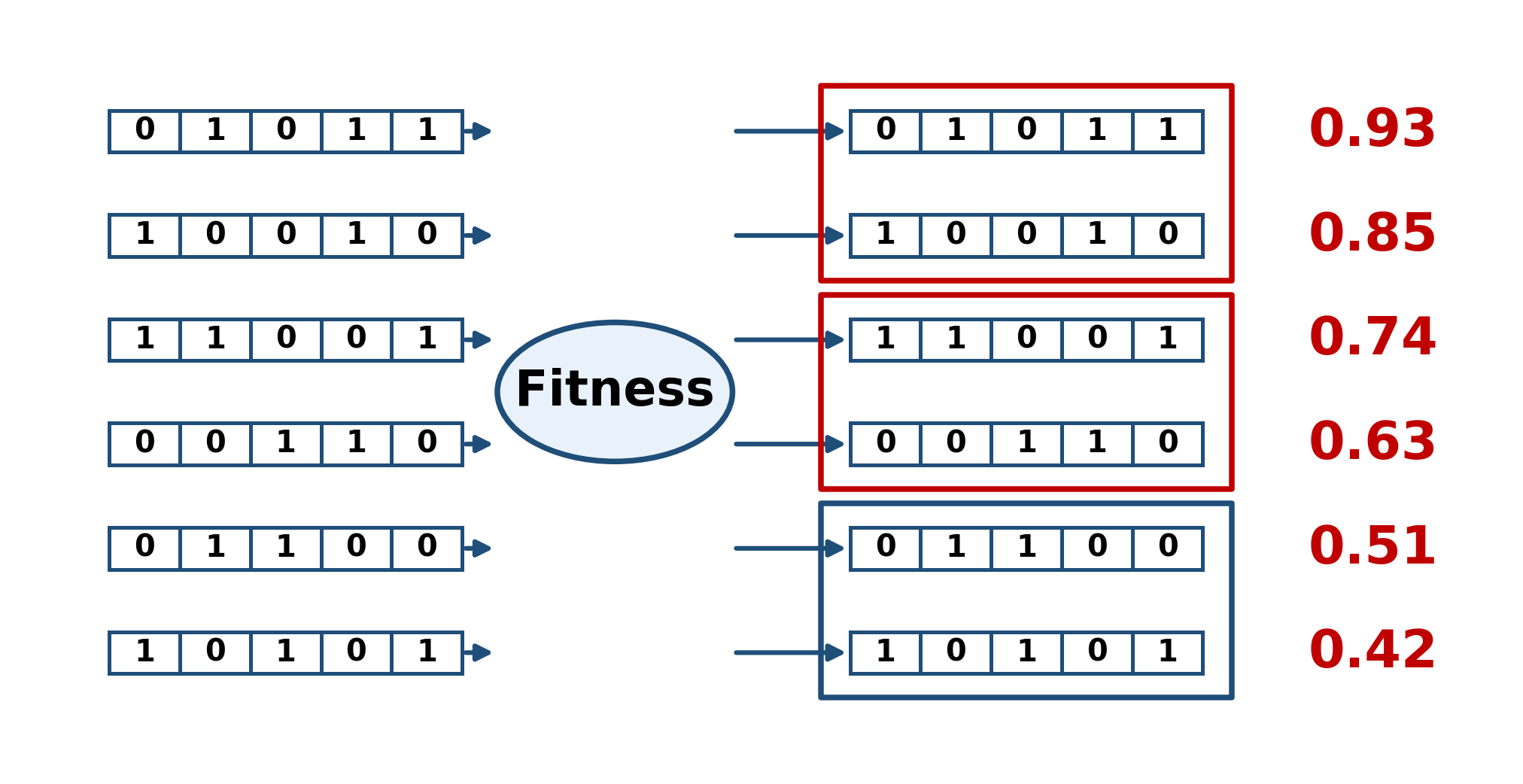}
        \caption{Fitness}
        \label{fig:sub2}
    \end{subfigure}

    \begin{subfigure}[b]{0.48\linewidth}
        \centering
        \includegraphics[width=1\linewidth]{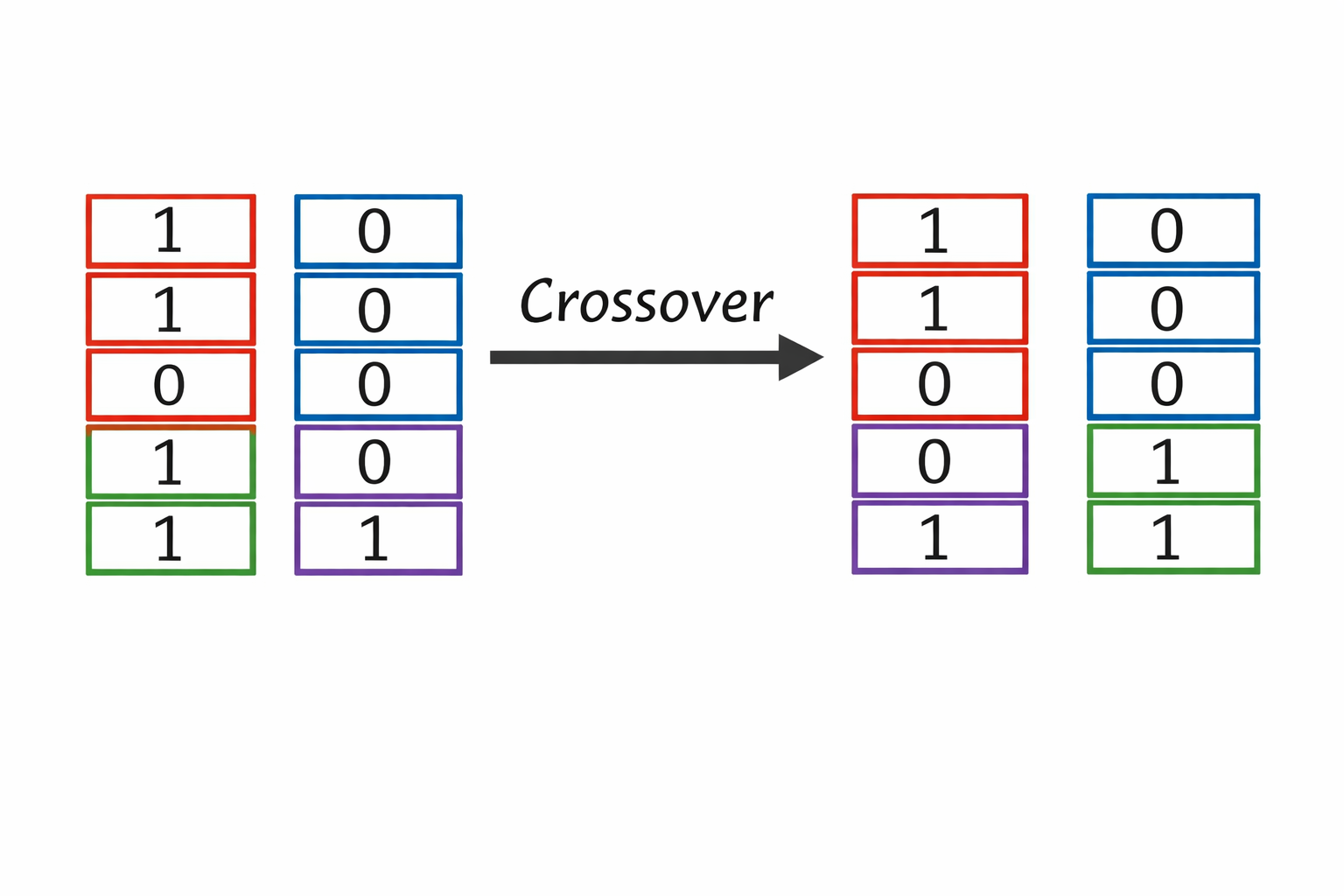}
        \caption{Crossover}
        \label{fig:sub3}
    \end{subfigure}
    \hfill
    \begin{subfigure}[b]{0.48\linewidth}
        \centering
        \includegraphics[width=1\linewidth]{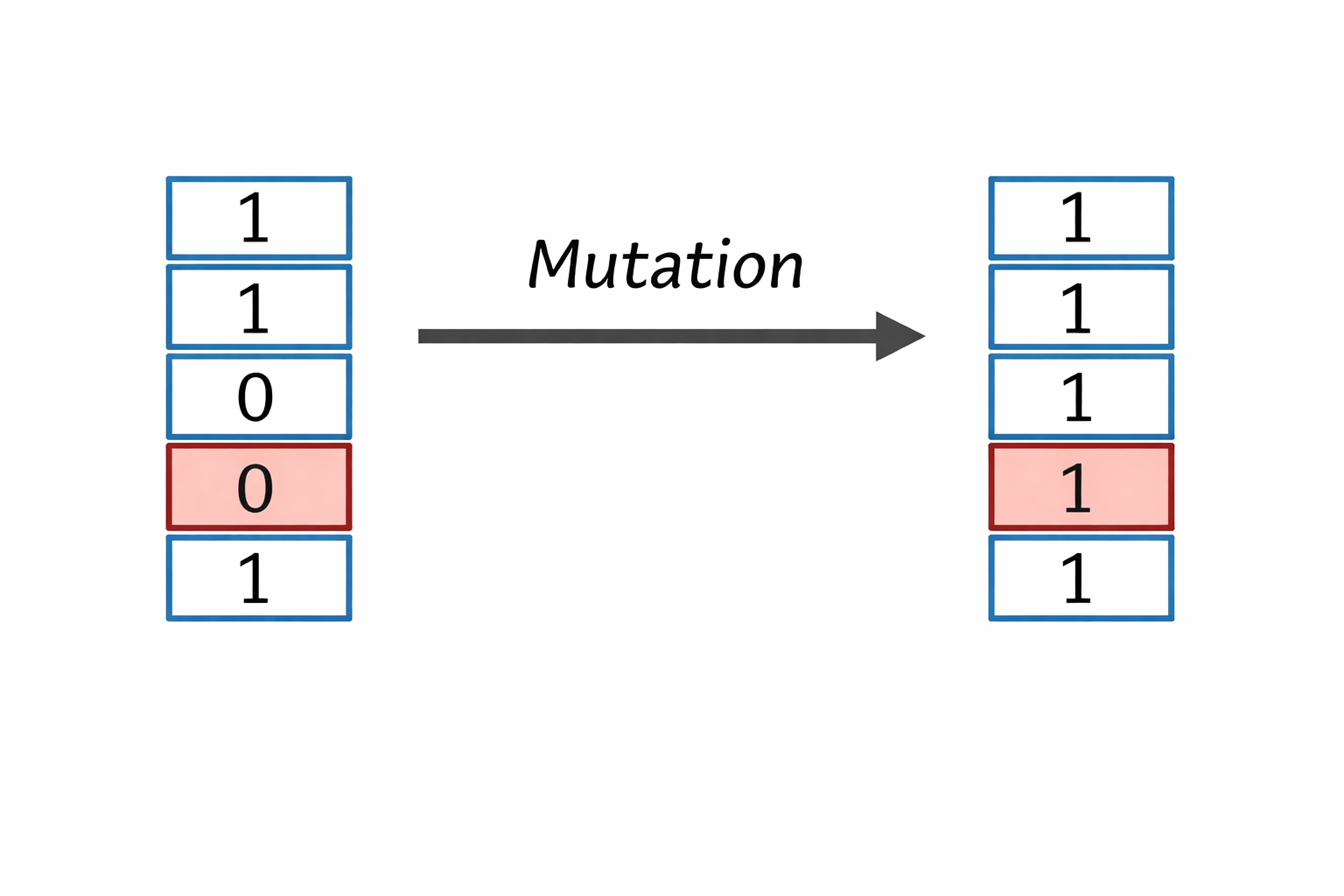}
        \caption{Mutation}
        \label{fig:sub4}
    \end{subfigure}

    \caption{Key components of the genetic algorithm: chromosome encoding, fitness evaluation, crossover, and mutation. The schematic summarizes how candidate solutions are represented, scored, recombined, and perturbed across generations.}
    \label{fig:GASTEPS}
\end{figure}

In practice, GA is often used as a global search technique, since the initial population is sampled from the entire parameter space. Its empirical performance generally improves as the population size increases and more generations are reproduced, although it comes with greater computational cost and complexity. Given a limited computational budget, small modifications to the selection, crossover, or mutation steps materially affect performance. With a sufficiently diverse initial population, GA often identifies regions close to a global maximizer within relatively few generations. However, the quality of later stage refinement depends strongly on the mutation mechanism .

To improve eventual convergence, a shrinkage is proposed into the mutation. Under this strategy, the region over which mutation is allowed shrinks as the number of generations increases. In practice, the shrinkage rate is often taken to be $\bigO(m^{-1/2})$, where $m$ denotes the generation index. This scaling is broadly consistent with convergence rate from classical optimization, as well with rates that arise in some statistical problems under independence or weak dependence.

Although mutation shrinkage can significantly improve GA performance in many settings, its limitations are more pronounced in high dimensional scenarios. Due to the curse of dimensionality, randomly initialized candidates tend to concentrate near the boundary of the parameter space, and a very large population are required to sufficiently explore its interior. Therefore, when the dimension is high and the population size is small or moderate, mutation shrinkage at rate $\bigO(m^{-1/2})$ may become too aggressive to leave substantial regions of the space unexplored.

\subsection{Challenges on Complexity}
We specialize the case in which $\partial f(\btheta,\Z)/\partial \btheta$ exists but $\partial f(\btheta,\Z)/\partial \Z$ does not. In our specific framework, GA treats $\Z$ as the chromosome and evaluates each candidate using the target function itself. Given a candidate value of $\widehat{\Z} = \Z$, we solve \eqref{easy} to obtain the corresponding optimizer with respect to $\btheta$. Iterating this process over successive generations produces a numerical approximation $(\widehat{\btheta}_{GA}, \widehat{\Z}_{GA})$ to the joint maximizer.

Suppose the population size in each generation is fixed as $n$, the GA is run for $m$ generations, and the computational complexity of evaluating $f(\btheta,\Z)$ is $\bigO(G)$. Ignoring the costs of crossover and mutation, the total complexity of the algorithm is $\bigO(nmG)$. Since $G$ is pre-determined by the structure of the target function, the dominant computational burden is driven primarily by the population size and the number of generations.

Empirically, the population size $m$ often has a greater effect on optimization accuracy than the number of generations. A natural way to improve precision is therefore to increase the population size. To illustrate the high-dimensional effect, consider the extreme case of a single generation search. Supposing $\Z$ is a $p\times 1$ parameter vector, the euclidean distance between the best sampled candidate $\widehat{\Z}_{GA}$ and the true optimizer $\widehat{\Z}$ is approximately Weibull distributed. We may write $||\widehat{\Z}_{GA} - \widehat{\Z}|| = \O_p(n^{-1/p})$, which implies the curse of high dimensionality. For fixed $n$, the approximation error deteriorates to $1$ rapidly as $p$ increases. Maintaining the same precision, therefore, requires population size to grow exponentially with dimension, which is computationally infeasible in high-dimensional problems.

Considering a small $n$, parts of the parameter space are unlikely to be explored at all, making it easy for the algorithm to miss the global maximizer. This problem is further exacerbated by the mutation step, which is restricted to a randomly centered cube whose side length shrinks at rate $\O(m^{-1/2})$ in generation $m$, since shrinkage progressively intensifies the omission.

\subsection{Improvement by Voronoi Partition}
Recent application-oriented studies have used Voronoi tessellation primarily as a compact representation of spatial or structural design, with GA serving as an outer-loop search routine. A recent review of Voronoi structures in materials science \citep{AlmeidaRios2025Review} shows applications in lightweight automotive, aerospace components, biomedical implants, and mechanically efficient cellular structures, which also describes GA-based optimization of bi-dimensional and triple-dimensional Voronoi structures as an emerging topic. In this work line, the Voronoi construction usually maps a relatively small set of design variables, such as seed locations and seed regularity, to a high-dimensional physical geometry. The fitness of each geometry is then evaluated by finite-element analysis, computational fluid dynamics, or surrogate prediction. For example, GA is used to improve the mechanical response of (quasi-)three-dimensional Voronoi porous structures \citep{Tung2023VoronoiGA,Lin2024ThreeDimVoronoiGA}, while \cite{Asakawa2024StiffenerVoronoi} optimize stiffener layouts in composite panels by placing stiffeners along Voronoi edges. \cite{FontalvoGarcia2025VoronoiRoofs} also apply it to determine the topology of Voronoi flat roofs with a varying number of seeds and structural member sizes, and \cite{Suzuki2025ArchitectedMaterialsInformatics} employ Voronoi cellular heat sinks through a neural-network surrogate model. Similar ideas have also appeared in deployment and coverage design, where Voronoi cells are used to prevent clustering and enforce segment-wise coverage in a GA search over pseudo-elite locations \citep{Xie2026Pseudolite}. These studies show the practical value of the combination, but in most cases Voronoi tessellation acts mainly as a domain-specific geometry generator or feasibility constraint, while the GA operators themselves remain largely conventional.

From a methodological perspective, the literature is more limited and tends to treat Voronoi partitioning as a modular auxiliary device for exploration, sampling, or representation. The Voronoi-established GA constructs models so that the offspring distribution can adapt to the objective landscape \citep{Shimosaka2004VMBGA}, and recent multi-objective learning work decomposes a high-dimensional design or preference space into Voronoi grids and applies GA to the grid-partitioning problem \citep{Chen2025VoronoiPFL}. In surrogate-assisted and adaptive sampling settings, PF-Voronoi sampling classifies regions containing predicted Pareto front points into Voronoi cells and selects new samples according to different criteria before applying Kriging \citep{Wu2024PFVoronoi}. In spatial data analysis, Voronoi-type methods use Voronoi cells and convex hull volumes as density units and then applies GA to identify cell blocks with low density \citep{Ghafour2025VEGA}. Thus, existing methods typically assign a separate role to each component that Voronoi partitions define local neighborhoods and regions, or geometric encodings, whereas GA supplies a generic global search engine, whose combination is more of multi-stage applications. 

The present paper is to integrate the two components more tightly. Instead of using a Voronoi partition only for initialization, representation, or spatial segmentation, we introduce the Voronoi cell as parts of the evolutionary mechanism itself. The graph partition controls how elite regions are recombined in crossover and how mutation neighborhoods are selected and shrunk. Consequently, the graph partition directly influences the offspring distribution, mutation refinements, the optimization trajectory, and ultimately reflects on the GA solution.

Motivated by these observations, we propose to represent the parameter space through a collection of non-overlapping convex cells and to use the induced cell adjacency graph to guide both exploration and genetic operations. Related ideas for generating new candidates appear in \cite{Richter2017} and \cite{Tian2009}. Although Voronoi partitions have considerable potential in global optimization, the literature on combining them with generally purpose optimization routines remains limited. For example, \cite{Shimosaka2004VMBGA} used Voronoi diagrams within a GA-type algorithm primarily at the crossover stage, whereas \cite{Tung2023VoronoiGA} combined Voronoi diagrams with GA to optimize mechanical properties of tension bearing structures. In this paper, we develop a flexible and general optimization framework that integrates Voronoi partitioning with GA. The proposed procedure could be viewed as a more systematic and broadly applicable variant of Basin hopping. Throughout, we assume that the objective function is continuous on $\mathcal{Z}$, since extending the approach to discrete valued $\Z$ would raise additional challenges and would not interact naturally with Voronoi partitioning.

The preceding discussion motivates an evolution operator that both retains high-fitness solutions and spreads exploration across geometrically distinct neighborhoods. Next section formalizes this idea through the proposed Voronoi-Elitism Genetic Algorithm.

\section{Optimization with Genetic Algorithm}

We propose a novel GA that improves crossover efficiency and robustness by incorporating geometric information extracted from high-quality solutions through a Voronoi-based partitioning. The proposed method develops from a GA framework and differs from conventional variants in its evolution operator. Its central innovation is a Voronoi based augmentation of the parent selection pool, which introduces structured local diversity during crossover while preserving emphasis on elite regions of the search space. To isolate the contribution of this design, we compare the proposed algorithm with two benchmark control variants as naive selection, denoted Control 0 (Random Mating) and Control 1 (Elite Weighted), which progressively incorporate standard mechanisms such as elitism and fitness-guided parent selection.

\subsection{Evolution Framework}
Starting from the initialized population, the algorithm proceeds iteratively by applying an evolution operator that updates the population from one generation to the next. In Algorithm \ref{AlgorithmEvolution}, the update is represented by a variant-specific \textsc{Evolve} procedure, which encapsulates parent selection, crossover, mutation, and any additional mechanisms. The same high-level control flow is used for all algorithmic variants. only the implementation of \textsc{Evolve} differs across methods.

Termination is determined by the relative improvement in the best fitness value between successive generations. Specifically, the algorithm stops when the best observed fitness fails to improve by more than a prespecified tolerance. In island-based implementations, the same criterion is applied to the best fitness across islands.

In the following subsections, we describe the evolution operators for Control 0, Control 1, and the proposed \textbf{Voronoi-Elitism GA (VEGA)}, emphasizing the incremental design changes that distinguish the three variants.

\textbf{Control 0} implements a minimal GA that serves as a baseline for evaluating more structured variants. It employs random parent selection, single point crossover, and additive Gaussian mutation at a given mutation rate, with neither elitism nor fitness-guided sampling. Thus, beyond the common fitness evaluation shared by all variants, Control 0 contains no additional structural mechanisms. See Algorithm \ref{Code:Control0} for pseudo code.

\begin{figure}[htbp]
    \centering
    \includegraphics[width=1\linewidth]{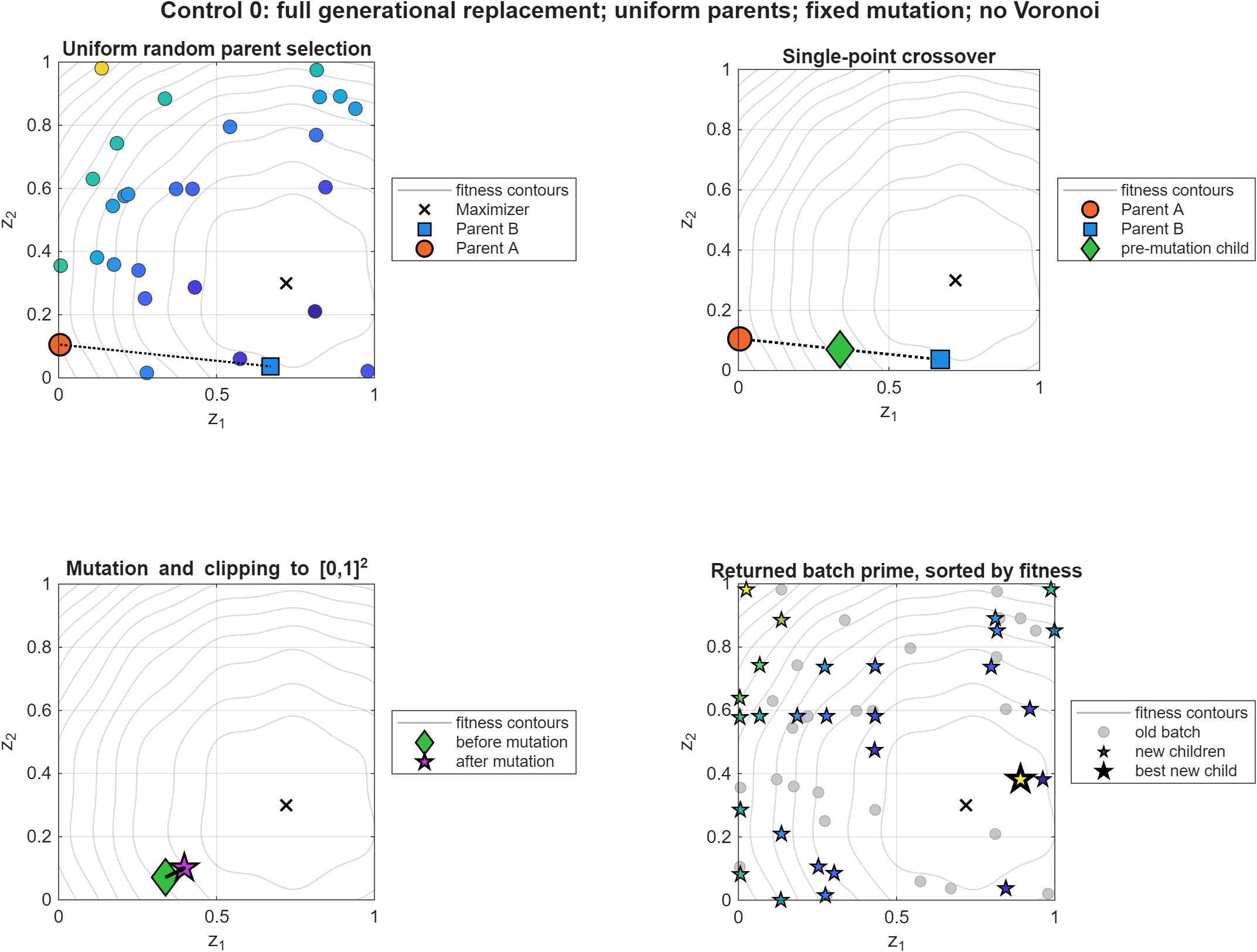}
    \caption{Control 0 evolution operator: parents are sampled uniformly from the current batch, recombined by single-point crossover, mutated with fixed-probability Gaussian noise, and returned as a fully replaced offspring batch sorted by fitness, with no Voronoi-based selection.}
    \label{fig:control0}
\end{figure}

\textbf{Control 1} extends the baseline algorithm in Control 0 by incorporating another two GA mechanisms: elitist retention and fitness-guided parent selection. In each generation, a fixed fraction $r \in (0,1)$ of the highest fitness individuals is copied directly into the next generation as elites. The remaining offspring are then generated from parents sampled randomly according to fitness, thereby preferring high-quality solutions while preserving randomness in the search process. See Algorithm \ref{Code:Control1} for pseudo code.

For Control 0, offspring are produced through single point crossover. Mutation is applied independently to each gene at a fixed mutation rate. Unlike in Control 0, however, the mutation magnitude is scaled by the difference between the two parent solutions and is bounded below by a minimum mutation threshold. The construction preserves a non-trivial exploratory step even when the selected parents are close to one another. After mutation, each offspring is projected back to $[0,1]$ by coordinate-wise clipping.

\begin{figure}[htbp]
    \centering
    \includegraphics[width=1\linewidth]{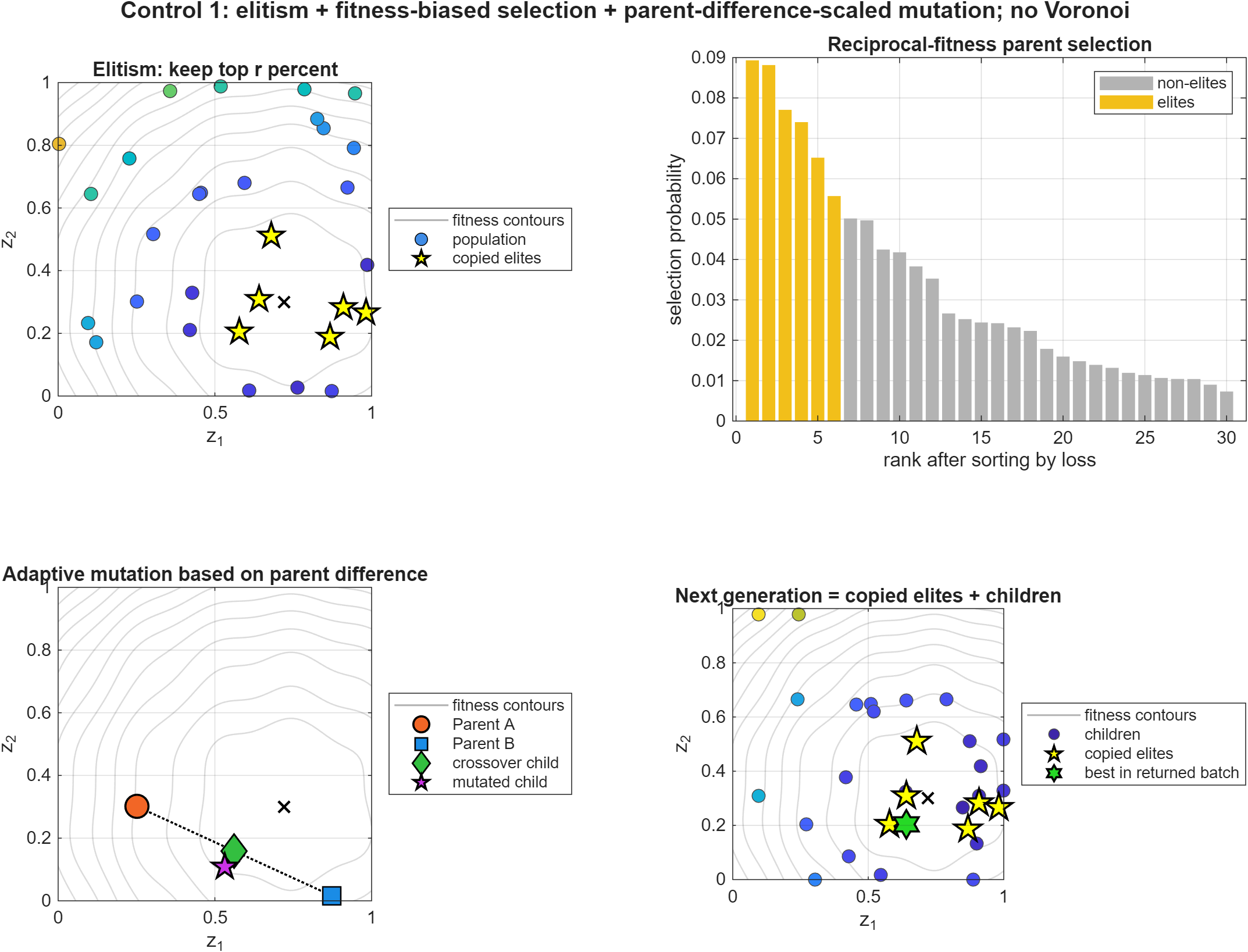}
    \caption{Control 1 evolution operator: top r(\%) elites are copied, while remaining children are generated from reciprocal-fitness-biased parents using single-point crossover and parent-difference-scaled mutation. The returned batch combines elites and children, sorted by fitness, without Voronoi selection.}
    \label{fig:control1}
\end{figure}

\subsection{Voronoi-Elitism Genetic Algorithm}

%{\color{cyan} Pick a good name, ``Voronoi-Augmented Evolution Operator"? or Voronoi-Augmented Genetic Algorithm? Try to distinguish from \cite{TUNG20233813} and \cite{SouzaAlmeida2025}, they focus on biological problem and we are somehow adaptive.}
%{\color{red} \fbox{Any comments about the name?}}
%{\color{magenta} We could name it ``Elite-Partition" or ``Voronoi-Augmented"}
%{\color{cyan}Maybe VEGA? Voronoi Elite GA? {\color{red} Does it sound cool (enough)?}}
%{\color{magenta} VEGA is fine I think}

The proposed method extends Control 1 by modifying the parent-selection mechanism used during crossover. For Control 1, a given elite fraction $r \in (0,1)$ is retained unchanged at each generation, and the remaining individuals are generated via crossover and mutation. The distinguishing feature of the proposed algorithm is the construction of an augmented parent pool, in which each elite solution is associated with locally generated candidates that reflect the geometric arrangement of elite points in the search space. This augmentation promotes structured local exploration while preserving emphasis on high-quality regions.

The Voronoi-based augmentation is implemented in a coordinate-wise, axis-aligned manner. For each dimension, the elite values are sorted, and adjacent midpoints are used to define interval boundaries. The Cartesian product of these intervals yields a hyper-rectangular region that approximates the local Voronoi cell associated with a given elite point. Notice that the exact Voronoi cells is also applicable, but it could be computationally expensive in high dimensional scenarios, including boundary identification and in augmentation generation. See Algorithm \ref{VPS} and Figures \ref{fig:VPS01} and \ref{fig:VPS02} for exact cases. To be convenient under limited budget, we focus on the rectangular approximation. A small number of candidate vectors is then sampled uniformly from this region. The first sampled candidate with finite fitness is retained. if no feasible candidate is obtained, the elite point itself is carried forward as a fallback.

The local candidates generated from the Voronoi construction are used to expand the parent pool for crossover. Specifically, reproduction draws from both the retained elites and their associated local samples, with parent selection with fitness-based probabilities. Offspring can inherit information from either the original elite solutions or nearby structured perturbations, which improves exploratory coverage preserving exploitation of high-fitness regions.

Voronoi construction is particularly robust in high dimensional settings, where a finite population often provides poor coverage of the parameter space, as random mutation with shrinkage may leave large regions unexplored. By generating local candidates from non-overlapping neighborhoods associated with distinct elite solutions, VEGA distributes exploratory effort systematically across promising parts and unexplored parts of the domain. In this sense, the method mitigates the tendency of the population to concentrate around a small subset of elites, but preserve potential geometric diversity in unexplored areas during crossover. See Algorithm \ref{Code:VEGA} for pseudo code.

\begin{figure}[htbp]
    \centering
    \includegraphics[width=1\linewidth]{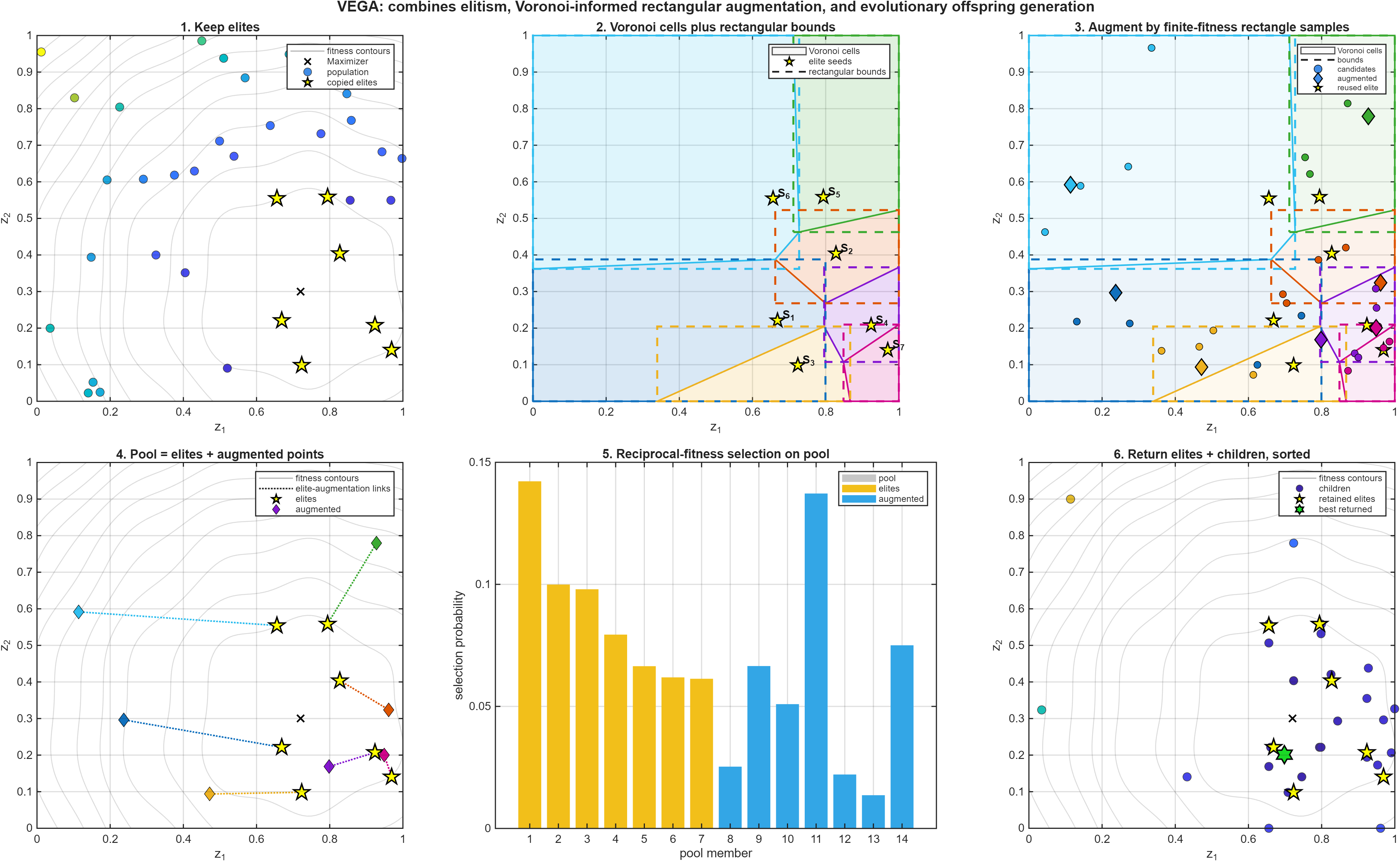}
    \caption{VEGA evolution operator: top elites are retained, each elite’s Voronoi-informed rectangular region is sampled to add finite-fitness augmented points, and elites plus augmented points form a fitness-biased parent pool. Offspring are generated by crossover and adaptive mutation, then returned with elites sorted by fitness.}
    \label{fig:VEGA}
\end{figure}

Next section discusses the theoretical performance of VEGA under limited budget. Section \ref{Section:Simulation} evaluates VEGA through simulation studies and comparative experiments against Control 0 and Control 1.

\section{Theoretical Results}
This section studies the convergence behavior of genetic algorithms under high dimensional settings. We analyze how population size, crossover, mutation, and the Voronoi mechanism affects the attainable distance. These results provide justification for the efficiency of VEGA under limited computational budgets, without requiring strong structural assumptions on the target function, while the asymptotic convergence of genetic algorithms is guaranteed by random perturbations theory \citep{Freidlin2012} using unlimited computational time \citep{Cerf1998}.

We begin with the initialization stage of a standard GA. Since the initial population is typically sampled uniformly from the search domain, the first step is to understand the geometric behavior of a single uniformly sampled chromosome in high dimension.

Suppose the high dimensional search area is $\bmS = [0, 1]^p$, and the chromosomes are sampled uniformly from $\bmS$. Consider the simplest case of single chromosomes $\X = (X_1, X_2, \cdots, X_p) \sim \operatorname{Unif}(\bmS)$. Notate the true maximizer as $\a = (a_1, a_2, \cdots, a_p) \in \bmS$ and the square distance between $\X$ and $\a$ could be denoted as $R^2=\|X-a\|^2=\sum_{i=1}^p\left(X_i-a_i\right)^2$.

\begin{lemma}
    Let $\X = (X_1, X_2, \cdots, X_p) \sim \operatorname{Unif}(\bmS)$, where $\bmS = [0, 1]^p$ is a high dimensional space, and $\a = (a_1, a_2, \cdots, a_p)  \subset \bmS$ but $\a \notin \partial \bmS$. Then $\X$ is asymptotically normal near a spherical shell centered at $\a$ with radius $r_p > 0$. Further for some $\sigma_0 > 0$,

    $$
    p^{-1}\| \X - \a\| \sim \mathcal{N}\left(r_p, \frac{\sigma_0^2}{p}\right),
    $$
    where $r_p = p^{-1} \sqrt{\sum_{i=1}^p\left[\frac{1}{12}+\left(a_i-\frac{1}{2}\right)^2\right]}$.
    \label{AsympNormalSample}
\end{lemma}

Lemma \ref{AsympNormalSample} shows that in high dimensions, a uniformly sampled chromosome is not spread evenly in terms of distance from $\a$. Instead, its distance concentrates around a thin spherical shell. This concentration provides a baseline for analyzing the initial performance of GA that, before selection or genetic operations are applied, most individuals are expected to lie at a similar distance from the true maximizer.

We, afterward, move from a single chromosome to an initial population of size $n$. Since selection in GA favors individuals with higher fitness, an idealized first generation estimator can be represented by the individual closest to the maximizer. Suppose it initially consists of $n$ individuals $\X_1, \X_2, \cdots, \X_n$, consisting the GA estimator at the first generation $\X_i$, for some $i = \argmin_{j}||\X_j - \a||^2$.

\begin{theorem}
    By notations and definitions in Lemma \ref{AsympNormalSample}, let $\X_1, \X_2, \cdots, \X_n \overset{i.i.d}{\sim} \operatorname{Unif}(\bmS)$. The distance between sample to $\a$, $R_i = \|\X_i - \a\|$, then the minimum distance satisfies
    $$R_{min} = r_p-\sigma_p \sqrt{2\log n}+\bigO\left(\frac{\log \log n}{\sqrt{\log n}}\right) + \bigO_p(\log^{-\frac{1}{2}}n).$$
    \label{MinimumDistance}
\end{theorem}

In Theorem \ref{MinimumDistance}, the negative term $-\sigma_p \sqrt{2\log n}$ captures the benefit of increasing the population size. As $n$ grows, the best individual among the initial samples is expected to be closer to $\a$, with the improvement occurring at the order of $\sqrt{\log n}$ under Gaussian approximation.

The previous result is derived for the unit hypercube. To extend the analysis beyond $\bmS = [0, 1]^p$, we consider a more general high dimensional domain and describe the distance behavior in terms of moment bounds.
Assume that the components of individual $\X$ satisfy $\E X_i^2 = \bigO(S^2)$, $\E X_i^4 = \bigO(B)$, with $\bigO(B)=\bigO(S^4)$. Hence, $R^2$ has expectation in \eqref{DistanceExpectationOrder} and variance in \eqref{DistanceVarianceOrder}.

\begin{equation}
    \E R^2 = \bigO(p S^2)+\sum_{i=1}^p\left(a_i-E X_i\right)^2 = \bigO\left(p S^2\right)
    \label{DistanceExpectationOrder}
\end{equation}

\begin{equation}
    \Var(R^2)=\bigO\left(p B+p S^4\right)
    \label{DistanceVarianceOrder}
\end{equation}

\begin{theorem}
    Let $\X_1, \X_2, \cdots, \X_n \overset{i.i.d}{\sim} \X$ for some uniform population $\X = (X_{1}, X_{2}, \cdots, X_{p})$ on a given non-degenerate high dimensional domain $\K$, with $p^{-1}\|\E \X\|$, $p^{-1}\|\E \X^2\| = \bigO(S^2)$, and $p^{-1}\|\E \X^4\| = \bigO(B)$ not relying on $p$. The distance between sample to some inner point $\a \subset \K$ but $\a \notin \partial \K$ is $R_i = \|\X_i - \a\|$, then the minimum distance satisfies
    $$\E R_{min}=\|\E \X - \a\|-\bigO(S^2 \log^{\frac{1}{2}} n)+\bigO_p(S \log^{-\frac{1}{2}} n).$$
    \label{MinimumDistanceOrderEst}
\end{theorem}

Having established the effect of population size at initialization, we further examine how crossover changes the distance distribution. We consider the simple arithmetic crossover operator $(X_1^{(2)}, X_2^{(2)}, \cdots, X_p^{(2)}) = \X^{(2)}=(\X_1+\X_2)/2$, where $\X_1$ and $\X_1$ are two independently sampled parent chromosomes. Therefore, the distance is defined as $R^2 = ||\X^{(2)} - \a||^2$. Similarly, its expectation and variance are 

\begin{equation}
        \E R^2 =\frac{p}{24}+\sum_{i=1}^p\left(a_i-\frac{1}{2}\right)^2=\sum_{i=1}^p\left(a_i-\frac{1}{2}\right)^2+\bigO\left(p S^2\right),
\end{equation}
and

\begin{equation}
        \Var(R^2)  =\sum_{i=1}^p \left\{E\left(X_i^{(2)}-a_i\right)^4-\left[\E\left(X_i^{(2)}-a_i\right)^2\right]^2\right\} =\bigO\left(p B+p S^4\right),
\end{equation}

Similar to Theorem \ref{MinimumDistance}, the minimum distance between GA estimator at the next generation and maximizer satisfies
\begin{equation}
    \begin{aligned}
R_{min} =& \mu_p+\sigma_p Z_{(1)}\\
=&\|\E \X - \a\|+\bigO\left(p^{\frac{1}{2}} S\right)-\bigO\left(p^{\frac{1}{2}} (B+S^4)^{\frac{1}{2}} \log^{\frac{1}{2}} n\right)\\
&+ \bigO_p\left((B+S^4)^{\frac{1}{4}} \log^{-\frac{1}{2}} n\right) \\
=&\sqrt{\sum_{i=1}^p\left(a_i-\frac{1}{2}\right)^2}-\bigO\left(S^2p^{\frac{1}{2}}\log^{\frac{1}{2}} n \right)+ \bigO_p\left(S \log^{-\frac{1}{2}} n\right).
\end{aligned}
\end{equation}

We then incorporate shrinking mutation. At generation $m$, component-wisely independent mutation is modeled as an additive random perturbation $\Z = (Z_1, Z_2,\cdots, Z_p) = \X + \W = (X_1 + W_1, X_2 + W_2, \cdots, X_p + W_p)$, where $\W = (W_1, W_2, \cdots, W_p)$ is independent of $\X$, satisfying $\E W_i = 0$, $\Var\left(W_i\right)=\bigO\left(m^{-1}\right)$, and $\E W_i^4 = \bigO\left(m^{-2}\right)$. This formulation captures the common practice of reducing the mutation magnitude as the number of generations increases.

Then distance is defined as $R^2 =\|\Z-\a\|^2=\|\X+\W-\a\|^2=\sum_{i=1}^p\left(X_i-a_i+W_i\right)^2$, with expectation and variance

\begin{equation}
    \begin{aligned}
        \E R^2 & = \E \sum_{i=1}^p\left(X_i^2+a_i^2+W_i^2-2 a_i X_i-2 a_i W_i+2 X_i W_i\right) \\
& =\sum_{i=1}^p Var\left(X_i\right)+\|\E \X - \a\|^2+\E W_i^2 \\
& =\sum_{i=1}^p\left(a_i-\frac{1}{2}\right)^2+\bigO\left(p S^2+\frac{p}{m}\right) \\
    \end{aligned}
\end{equation}

\begin{equation}
\begin{aligned}
    \Var R^2 &=\sum_{i=1}^p \left\{\E\left(X_i+W_i-a_i\right)^4-\left[\E\left(X_i+W_i-a_i\right)^2\right]^2\right\} \\
    & =\bigO\left(p\left(S^4+\frac{1}{m}\right)\right).
\end{aligned}
\end{equation}

Using arguments in previous proofs, the minimum of distance satisfies

\begin{equation}
    \begin{aligned}
R_{min} =& \mu_p+\sigma_p Z_{(1)}\\
=&\|\E \X - \a\|+\bigO\left((p S^2 + pm^{-1})^{\frac{1}{2}}\right)-\bigO\left(p^{\frac{1}{2}} (S^4 + m^{-1})^{\frac{1}{2}} \log^{\frac{1}{2}} n\right)\\
&+ \bigO_p\left((S^4+m^{-1})^{\frac{1}{4}} \log^{-\frac{1}{2}} n\right) \\
=&\sqrt{\sum_{i=1}^p\left(a_i-\frac{1}{2}\right)^2}-\bigO\left((S^2 + m^{-\frac{1}{2}})p^{\frac{1}{2}}\log^{\frac{1}{2}} n \right)\\
&+ \bigO_p\left((S + m^{-\frac{1}{4}}) \log^{-\frac{1}{2}} n\right).
\end{aligned}
\end{equation}

In shrinkage mutation implement, mutation errors are accumulative. It is easy to have

\begin{equation}
    S^2=\bigO\left(\sum_{i=1}^m \frac{1}{i}\right)=\bigO(\log m).
\end{equation}

Hence
\begin{equation}
    R_{min}=\sqrt{\sum_{i=1}^p\left(a_i-\frac{1}{2}\right)^2}-\bigO(p^{\frac{1}{2}}\log^{\frac{1}{2}} m \log^{\frac{1}{2}} n)+\bigO_p(\log^{\frac{1}{2}} n \log^{\frac{1}{2}} m).
    \label{RminUnderShrinkingMutation}
\end{equation}

Equation \eqref{RminUnderShrinkingMutation} suggests that the accumulated effect of shrinking mutation can enlarge the search dispersion over generations. In high dimensions, this additional dispersion may improve the chance of producing individuals closer to the maximizer, provided that the mutation scale is controlled so that the approximation remains valid under limited budget.

To analyze the Voronoi mechanism used in VEGA, we introduce order notation and a geometric bound for uniform samples over a high-dimensional subset of the unit cube. We use the notation $A\lesssim B$ to show that $A\le cB$ for some universal constant $c>0$. We also denote $A\asymp B$ if $A\lesssim B$ and $B\lesssim A$.

\begin{theorem}[Order bounds]
Let $\K\subseteq[0,1]^p$ with $|\K|=v$, where $0<v\le1$, and let $\X\sim\operatorname{Unif}(\K)$. Then
$$
\sqrt p\,v^{1/p} \lesssim \E\|\X\| \lesssim \sqrt p,
$$
and
$$
\Var\|\X\| \lesssim \bigO(p) - \bigO\!\left(pv^{1/p}\right) - \bigO\!\left(pv^{2/p}\right).
$$
\label{Theorem:OrderBounds}
\end{theorem}

Finally, we incorporate the Voronoi partition strategy of VEGA. Unlike standard GA, which relies only on genetic operations within the existing population, VEGA partitions the feasible domain and encourages exploration across Voronoi cells. This mechanism effectively increases spatial coverage, which can improve the order of the best attainable distance in high dimensions.

\begin{theorem}
    $\K \subset [0,1]^p$ is a non-degenerate high-dimensional convex set with volume $|\K| \in [0, 1]$, and samples $\X_1, \X_2, \cdots, \X_n \overset{i.i.d}{\sim} \operatorname{Unif}(K) + \bm{\varepsilon}$, where $\bm{\varepsilon} \subset \mathbb{R}^p$ is an independent random variable with $\|\E \bm{\varepsilon}\| = 0$ and $\|\Var(\bm{\varepsilon})\| = \bigO(pg(m))$. Let $\C_1(\K), \C_2(\K), \ldots, \C_{[rn]}(\K)$ is Voronoi Partition of $\K$, which is applied to VEGA in Algorithm \ref{VPS}. Then the minimum distance of $m$-generation VEGA estimator $\X^{(m)}$ to true value $\a \notin \partial \K$, for some small $\xi$, $R_{min}^{(m)} = \|\X^{(m)} - \a\|$ satisfies
    \begin{equation*}
        \begin{aligned}
            R_{min}^{(m)} \lesssim & \bigO\left(p^{\frac{1}{2}}n^{-\frac{(1+\xi)}{p}}-p^{\frac{1}{2}}g(m)\log^{\frac{1}{2}} (n + mrn)\right)\\
            &+\bigO_p\left(g^{\frac{1}{2}}(m) \log^{-\frac{1}{2}} (n + mrn)\right).
        \end{aligned}
    \end{equation*}

\label{VEGAConsistency}
\end{theorem}

Theorem \ref{VEGAConsistency} suggests that the Voronoi mechanism improves the spatial coverage of the search domain and can yield a sharper order of the best distance based only on crossover or shrinking mutation. This advantage is particularly relevant high dimensional settings under limited budget, where increasing the population size or the number of generations alone may be insufficient to overcome distance concentration.

\section{Simulation Studies: Segmented Regression} \label{Section:Simulation}

VEGA is general and can be applied to a wide range of optimization problems through the specification of an appropriate fitness function and, when necessary, a problem-specific population representation. In this work, we instantiate the algorithm for segmented regression by defining a custom fitness function and a population initialization procedure that enforces the ordering constraints required for breakpoint representations.

\subsection{Data Generation and Simulation Setting}
\paragraph{Synthetic Data Generation.} Given $K$ segments, we generate and design a predefined linear model piece-wisely, which is well-conditioned and non-degenerate. Breakpoint locations are sampled randomly over the domain subject to ordering and minimum separation constraints, ensuring that all segments have comparable width to each other.

Specific segment slopes and intercepts are then generated sequentially. The parameters for the first segment are drawn at random, and parameters for following segments are sampled with a minimum difference from the preceding segment. This constraint enforces a clear structural change at each breakpoint, either through a change in slope or through a discontinuity in the fitted value at the segment boundary. As a result, adjacent segments are guaranteed to be meaningfully distinct, avoiding cases in which breakpoints correspond to only negligible changes in the underlying signal.

Given the resulting segmentation and segment parameters, the response values are generated according to a piecewise linear model with additive Gaussian noise. Specifically, for each observation $i$,
$$
y_i = \theta_{s(i)}\, x_i + \alpha_{s(i)} + \varepsilon_i, \qquad 
\varepsilon_i \overset{i.i.d}{\sim} \mathcal{N}(0,\sigma^2),
$$
where $x_i$ denotes the covariates, $y_i$ the corresponding response, and $s(i)\in\{1,\dots,K\}$ the segment index determined by the true breakpoints. The parameters $\theta_{s(i)}$ and $\alpha_{s(i)}$ denote the slope and intercept of segment $s(i)$, respectively.

\paragraph{Population Initialization.} All GA candidates share the same population initialization procedure. Each individual is represented by a sorted breakpoint vector $\Z \in [0,1]^{K-1}$. Initial candidates are generated by uniformly sampling breakpoint values in $[0,1]$, and afterward evaluated and sorted according to the fitness of segmentation results. Candidates with infinite fitness are dropped, and sampling is repeated until the desired population size is reached. This rejection-driven initialization ensures that all individuals in the initial population correspond to valid segmentations while introducing no additional structural bias.

\paragraph{Fitness Evaluation.} Candidate solutions are encoded as breakpoint vectors $Z \in [0,1]^{K-1}$, which induces $K$ contiguous segments by partitioning the input variable $X$. For given $\Z$, segment boundaries are first determined, and the data are divided accordingly into specific segment subsets $(\X_s, \Y_s)$ for $s = 1,\dots,K$. Fitness is computed by fitting an independent Ordinary Least Squares (OLS) linear model within each segment and aggregating the squared residuals. If any segment is empty or fails to meet minimum support requirements, the candidate solution is deemed invalid and assigned with infinite fitness. For valid segmentations, the total loss is calculated as the sum of squared residuals across all segments, and the final fitness value is obtained by normalizing this aggregated loss by the sample size. This fitness definition is used consistently across all algorithmic variants. See Algorithm \ref{Code:SimulationFitness} for pseudo code.

\begin{algorithm}[htbp]
\caption{\textbf{Fitness Evaluation}}
\begin{algorithmic}[1]
\Procedure{\texttt{Fitness}}{$X,y,Z,K$}
    \State Compute segment index boundaries induced by $Z$ on sorted $X$
    \State loss $\gets 0$
    \For{$s = 1$ to $K$}
        \State Extract segment $(X_s, y_s)$
        \If{segment is empty or too narrow}
            \State \Return $\infty$
        \EndIf
        \State Fit $(\beta_s,\alpha_s)$ on $(X_s, y_s)$ by OLS
        \State loss $\gets$ loss $+ \sum_{(x,y)\in(X_s,y_s)} (\beta_s x + \alpha_s - y)^2$
    \EndFor
    \State \Return $\text{loss}/|X|$
\EndProcedure
\end{algorithmic}
\label{Code:SimulationFitness}
\end{algorithm}

\subsection{Simulation Results}

We evaluate the VEGA by comparing it with Control 0 and Control 1. Although the algorithm itself is general, all experiments are conducted in a segmented regression setting, which provides a clear and structured context for comparing optimization behavior.

Algorithm performance is evaluated under three criteria. First, we assess prediction accuracy by Mean Squared Error (MSE) given known number of segments. Second, we examine breakpoint recovery to evaluate how accurately each method identifies the underlying segmentation structure under simplified conditions. Finally, we study model selection performance by Bayesian Information Criterion (BIC), assuming the number of segments is unknown, with particular emphasis on the ability of proposed algorithm to recover the correct model complexity.

\subsubsection{MSE Performance}

Prediction performance is evaluated using synthetic data generated from piecewise linear regression models with known segmentation structure. Experiments are conducted for segment counts $K \in \{4,8\}$, segment lengths in $\{100,300,1000\}$, and noise levels $\sigma \in \{0.1,0.2,0.4\}$. For each configuration $(K, n, \sigma)$, breakpoint locations and corresponding segment slopes and intercepts are randomly generated and held fixed, defining the underlying signal. Covariates are taken on an evenly spaced grid over $[0,1]$, with total sample size $n = K \times \text{segment length}$.

For each simulation replicate, Gaussian noise is independently generated and added to the fixed signal to obtain observed responses. Given the true segmentation, segment-wise regression parameters are refit via ordinary least squares (OLS) to account for the variability introduced by noise, providing a baseline corresponding to the best achievable fit under the true model.

Prediction performance is evaluated using mean squared error (MSE). For each candidate model, breakpoint locations are estimated by Control 0, Control 1, and VEGA under Simple and Island settings, and segment-wise parameters are subsequently refit via OLS given the estimated segmentation. For all GA-based methods, we use a population size of $50K$, a maximum of $2K$ iterations, and an island model with 5 islands, where crossover between islands occurs every 2 generations for $K=4$ and every 4 generations for $K=8$. See Table \ref{Table:MSE} for results.

\begin{table}[htpb]
\centering
\scriptsize
\caption{Average Bias and MSE for $K=4,8$ across segment lengths ($l$) and noise levels ($\sigma$).}
\begin{tabular}{ccc|rr|rr|rr}
 &              &       & \multicolumn{2}{c|}{Control 0} & \multicolumn{2}{c|}{Control 1} & \multicolumn{2}{c}{VEGA} \\
$K$ &$l$            & $\sigma$ & Bias$\times 10^3$ & MSE$\times 10$ & Bias$\times 10^3$ & MSE$\times 10$ & Bias$\times 10^3$ & MSE$\times 10$ \\ \hline
\addlinespace[2pt]
\multicolumn{9}{c}{Simple GA} \\
\addlinespace[2pt] \hline
4 &  100 & 0.1 & 30.054 & 0.399 &  4.799 & 0.146 &  0.184 & 0.100 \\
  &      & 0.2 & 29.922 & 0.692 &  6.419 & 0.457 &  0.111 & 0.394 \\
  &      & 0.4 & 27.515 & 1.844 &  7.513 & 1.644 & -0.381 & 1.565 \\
  &  300 & 0.1 & 40.669 & 0.506 &  6.786 & 0.167 &  0.864 & 0.108 \\
  &      & 0.2 & 39.797 & 0.796 &  8.778 & 0.486 &  0.694 & 0.405 \\
  &      & 0.4 & 39.038 & 1.979 & 10.794 & 1.696 &  0.663 & 1.595 \\
  & 1000 & 0.1 & 30.503 & 0.405 &  5.459 & 0.155 &  0.377 & 0.104 \\
  &      & 0.2 & 30.294 & 0.703 &  7.500 & 0.475 &  0.359 & 0.403 \\
  &      & 0.4 & 29.884 & 1.896 &  9.573 & 1.693 &  0.274 & 1.600 \\
8 &  100 & 0.1 & 56.867 & 0.666 & 11.805 & 0.216 &  1.702 & 0.115 \\
  &      & 0.2 & 56.689 & 0.958 & 13.190 & 0.523 &  1.668 & 0.408 \\
  &      & 0.4 & 55.164 & 2.117 & 14.744 & 1.712 &  0.903 & 1.574 \\
  &  300 & 0.1 & 52.762 & 0.627 & 14.386 & 0.243 &  3.126 & 0.131 \\
  &      & 0.2 & 52.668 & 0.924 & 16.484 & 0.562 &  3.185 & 0.429 \\
  &      & 0.4 & 52.330 & 2.114 & 18.611 & 1.777 &  2.748 & 1.618 \\
  & 1000 & 0.1 & 49.616 & 0.596 & 12.902 & 0.229 &  3.671 & 0.137 \\
  &      & 0.2 & 49.040 & 0.889 & 14.407 & 0.543 &  3.695 & 0.436 \\
  &      & 0.4 & 49.551 & 2.092 & 16.579 & 1.762 &  3.528 & 1.632 \\ \hline
\addlinespace[2pt]
\multicolumn{9}{c}{Island GA} \\
\addlinespace[2pt] \hline
4 &  100 & 0.1 & 30.388 & 0.402 &  4.682 & 0.145 &  0.201 & 0.100 \\
  &      & 0.2 & 29.782 & 0.689 &  6.555 & 0.456 &  0.130 & 0.392 \\
  &      & 0.4 & 27.794 & 1.845 &  7.773 & 1.645 & -0.376 & 1.563 \\
  &  300 & 0.1 & 39.455 & 0.494 &  6.772 & 0.167 &  0.788 & 0.107 \\
  &      & 0.2 & 39.638 & 0.793 &  8.320 & 0.480 &  0.996 & 0.407 \\
  &      & 0.4 & 38.873 & 1.977 & 10.601 & 1.695 &  0.476 & 1.593 \\
  & 1000 & 0.1 & 30.476 & 0.404 &  5.544 & 0.155 &  0.402 & 0.104 \\
  &      & 0.2 & 30.457 & 0.704 &  7.166 & 0.471 &  0.349 & 0.403 \\
  &      & 0.4 & 29.847 & 1.895 &  9.322 & 1.690 &  0.310 & 1.600 \\
8 &  100 & 0.1 & 56.323 & 0.661 & 11.574 & 0.214 &  1.796 & 0.116 \\
  &      & 0.2 & 56.299 & 0.955 & 13.523 & 0.527 &  1.780 & 0.410 \\
  &      & 0.4 & 55.112 & 2.117 & 14.785 & 1.713 &  0.876 & 1.574 \\
  &  300 & 0.1 & 52.447 & 0.624 & 14.728 & 0.246 &  3.158 & 0.131 \\
  &      & 0.2 & 52.862 & 0.926 & 16.495 & 0.563 &  3.319 & 0.431 \\
  &      & 0.4 & 52.220 & 2.112 & 18.551 & 1.775 &  2.827 & 1.618 \\
  & 1000 & 0.1 & 49.495 & 0.595 & 12.958 & 0.229 &  3.554 & 0.135 \\
  &      & 0.2 & 49.441 & 0.893 & 14.555 & 0.545 &  3.601 & 0.435 \\
  &      & 0.4 & 49.229 & 2.089 & 16.770 & 1.765 &  3.513 & 1.632 \\ \hline
\end{tabular}
\label{Table:MSE}
\end{table}

\subsubsection{Breakpoint Recovery}

To evaluate the ability of each algorithm to recover structural information beyond prediction accuracy, we conduct breakpoint recovery experiments in a simplified setting with a single interior breakpoint. The number of segments is fixed at $K = 2$, with the true breakpoint located at $0.5$. This setup isolates the task of breakpoint estimation and allows direct comparison between estimated and true breakpoint locations.

Segment slopes and intercepts are independently sampled from the uniform distribution on $(-1,1)$ and held fixed across all simulation replicates. To ensure identifiability, the absolute difference between the slopes of the two segments is constrained to be at least $0.5$, guaranteeing a meaningful structural change at the breakpoint. Covariates are generated on an evenly spaced grid over $[0,1]$, with sample sizes $n \in \{1000,2000,4000\}$. Observed responses are obtained by evaluating the corresponding piecewise linear model and adding i.i.d.\ Gaussian noise with standard deviation $\sigma \in \{0.1,0.2,0.4\}$. For each configuration, 1,000 independent simulation replicates are generated.

For model fitting, all methods are initialized with a population of size 50, with candidate breakpoint locations sampled uniformly from $\{0,1\}$ to mitigate the effects of random initialization in this low-dimensional setting. Each algorithm is run for a maximum of 10 iterations, and for island-based methods, crossover between islands occurs every 3 generations. Breakpoint estimates are obtained from Control 0, Control 1, and VEGA under both Simple and Island configurations. See Table \ref{Table:Breakpoint} for results and Figure \ref{fig:Breakpoints} for breakpoint recovery distribution, which aligns with the theoretical results about non-standard asymptotic distribution of breakpoint estimators in \cite{Ling2016,Ma2016}.

\begin{table}[htpb]
\centering
%{\fontsize{8pt}{10pt}\selectfont
\scriptsize
\caption{SD and MSE of breakpoint estimates across segment lengths $l$ and noise levels $\sigma$, computed relative to the true breakpoint $\tau=0.5$.}
\begin{tabular}{cc|rr|rr|rr}
               &       & \multicolumn{2}{c|}{Control 0} & \multicolumn{2}{c|}{Control 1} & \multicolumn{2}{c}{VEGA} \\
$l$            & $\sigma$ & SD$\times 10$ & MSE$\times 10^2$ & SD$\times 10$ & MSE$\times 10^2$ & SD$\times 10$ & MSE$\times 10^2$  \\ \hline 
\addlinespace[2pt]
\multicolumn{8}{c}{Simple GA} \\
\addlinespace[2pt] \hline
1000 & 0.1 & 3.08 & 10.21 & 0.89 &  0.80 & 0.14 &  0.02 \\
     & 0.2 & 3.03 &  9.76 & 0.75 &  0.56 & 0.33 &  0.11 \\
     & 0.4 & 3.18 & 10.30 & 0.74 &  0.55 & 0.57 &  0.32 \\
2000 & 0.1 & 3.08 &  9.90 & 0.98 &  0.96 & 0.15 &  0.02 \\
     & 0.2 & 3.06 &  9.99 & 0.77 &  0.60 & 0.18 &  0.03 \\
     & 0.4 & 3.06 &  9.65 & 0.69 &  0.47 & 0.32 &  0.10 \\
4000 & 0.1 & 3.14 & 10.14 & 0.96 &  0.92 & 0.10 &  0.01 \\
     & 0.2 & 3.06 &  9.85 & 0.75 &  0.57 & 0.15 &  0.02 \\
     & 0.4 & 3.03 &  9.65 & 0.66 &  0.43 & 0.28 &  0.08 \\ \hline
\addlinespace[2pt]
\multicolumn{8}{c}{Island GA} \\
\addlinespace[2pt] \hline
1000 & 0.1 & 3.11 &  9.92 & 0.92 &  0.85 & 0.16 &  0.03 \\
     & 0.2 & 3.10 &  9.98 & 0.84 &  0.70 & 0.27 &  0.08 \\
     & 0.4 & 3.14 & 10.23 & 0.78 &  0.61 & 0.55 &  0.30 \\
2000 & 0.1 & 3.12 & 10.10 & 0.91 &  0.83 & 0.13 &  0.02 \\
     & 0.2 & 3.09 &  9.88 & 0.79 &  0.63 & 0.19 &  0.04 \\
     & 0.4 & 3.05 &  9.82 & 0.70 &  0.50 & 0.37 &  0.14 \\
4000 & 0.1 & 3.08 & 10.06 & 0.96 &  0.92 & 0.10 &  0.01 \\
     & 0.2 & 3.18 & 10.54 & 0.81 &  0.66 & 0.17 &  0.03 \\
     & 0.4 & 3.12 & 10.15 & 0.60 &  0.36 & 0.27 &  0.07 \\
\end{tabular}
\label{Table:Breakpoint}

\end{table}

\begin{figure}[htbp]
    \centering
    \includegraphics[width=0.95\linewidth]{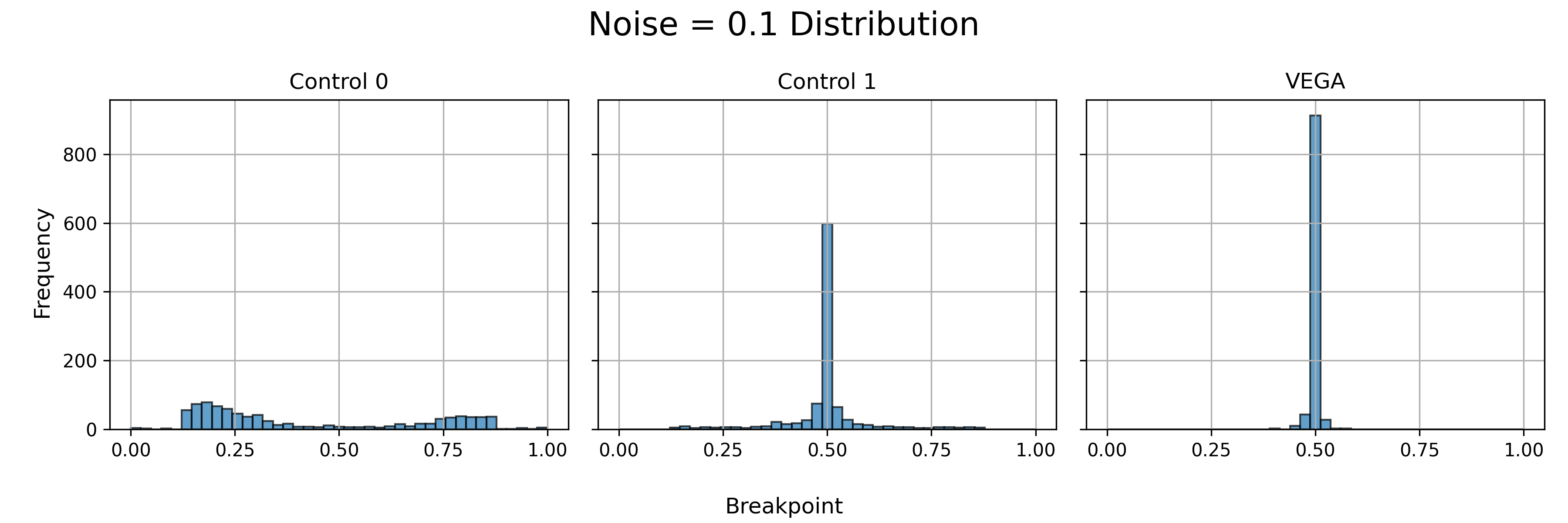}
    \includegraphics[width=0.95\linewidth]{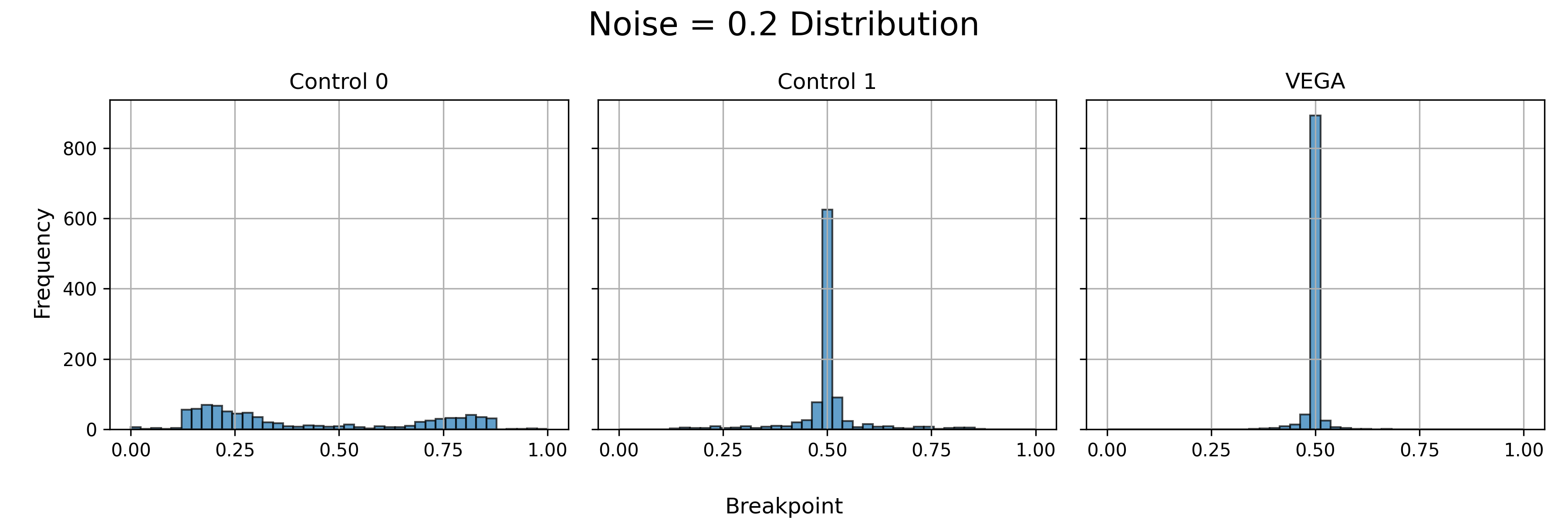}
    \includegraphics[width=0.95\linewidth]{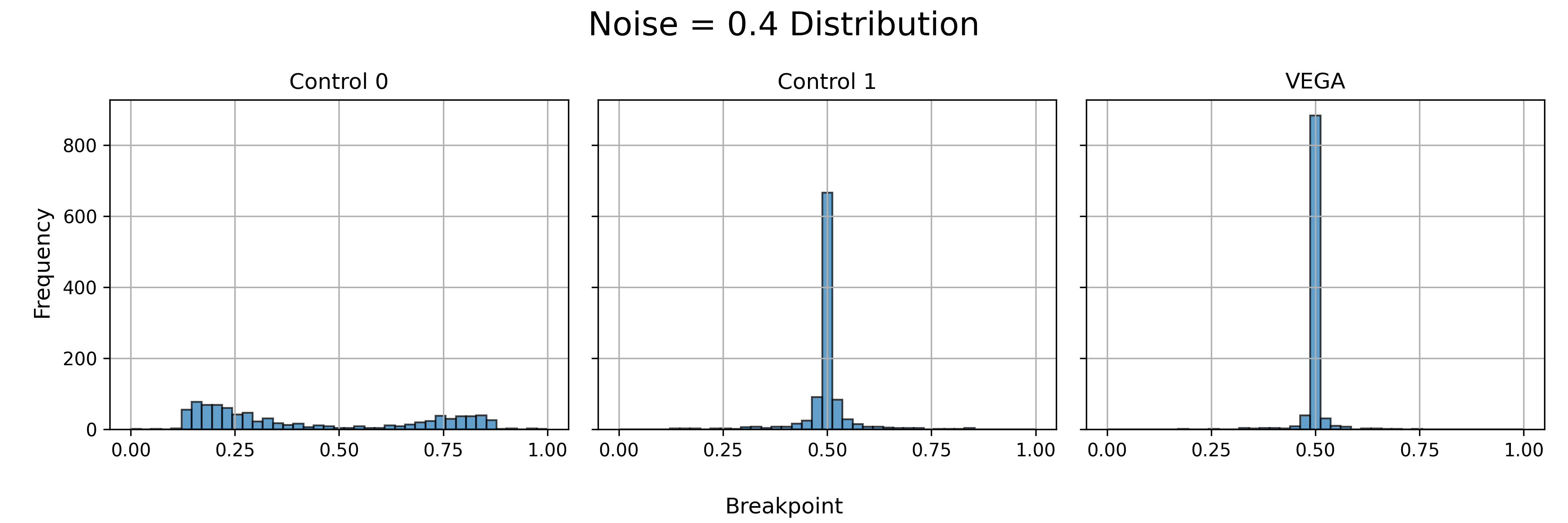}
    \caption{Breakpoint distribution for the three models at $l=4000$ and $\sigma=0.1, 0.2, 0.4$}
    \label{fig:Breakpoints}
\end{figure}

\section{Real Data Analysis}

We apply VEGA to real-world data from Stack Exchange to study long-term activity patterns and structural changes over time. The dataset consists of post and vote counts aggregated over approximately 5.5 years, from February 8, 2018 to July 23, 2023. We focus on three primary measures of platform engagement: 
\begin{itemize}[noitemsep, topsep=0pt]
    \item \textbf{Questions}: Number of questions posted per trading day,
    \item \textbf{Answers}: Number of answers posted per trading day, and
    \item \textbf{Accept Votes}: Number of accepted-answer votes per trading day, where an accepted-answer vote corresponds to a question author marking an answer as accepted. 
\end{itemize}
These variables capture complementary aspects of content generation and resolution dynamics on the platform.

Daily activity exhibits substantial variability; however, all model fitting and breakpoint estimation are performed directly on the original data. We fit segmented regression models using each genetic-algorithm variant for candidate segment counts ranging from $K = 2$ to $10$. For each value of $K$, the algorithm is run five times with different random initializations, and the solution with the lowest mean squared error (MSE) is retained.

Model selection is conducted using the Bayesian Information Criterion (BIC), and the optimal number of segments is chosen by minimizing the BIC value. Across all three measures of platform engagement, the selected model consistently yields $K = 4$ segments, indicating the presence of four major regimes in the underlying activity patterns.

\subsection{Residual Bootstrapping}

To assess the stability and uncertainty of the estimated breakpoints on real data, we apply a residual-based bootstrap that preserves the segmented structure in the fitted model. After selecting the best model and its breakpoints using BIC on the training series, we compute segment-wise fitted values and residuals from the original data. Bootstrap datasets are then generated by independently shuffling the residuals within each fitted segment and adding them back to the concatenated fitted values, thereby retaining the estimated segment shapes while resampling only the unexplained variation. For each bootstrap replicate (5000 total), we re-run the genetic algorithm with the same number of segments $K$ to obtain a new set of breakpoint estimates. The resulting collection of bootstrap breakpoint vectors is used to evaluate stability and uncertainty by examining the empirical distribution of breakpoint locations, constructing percentile-based confidence intervals, and measuring the fraction of replicates that place a breakpoint within a small window of the original estimate, reflecting how consistently the algorithm recovers similar segmentation structure.

\begin{figure}[htbp]
    \centering
    \includegraphics[width=0.8\linewidth]{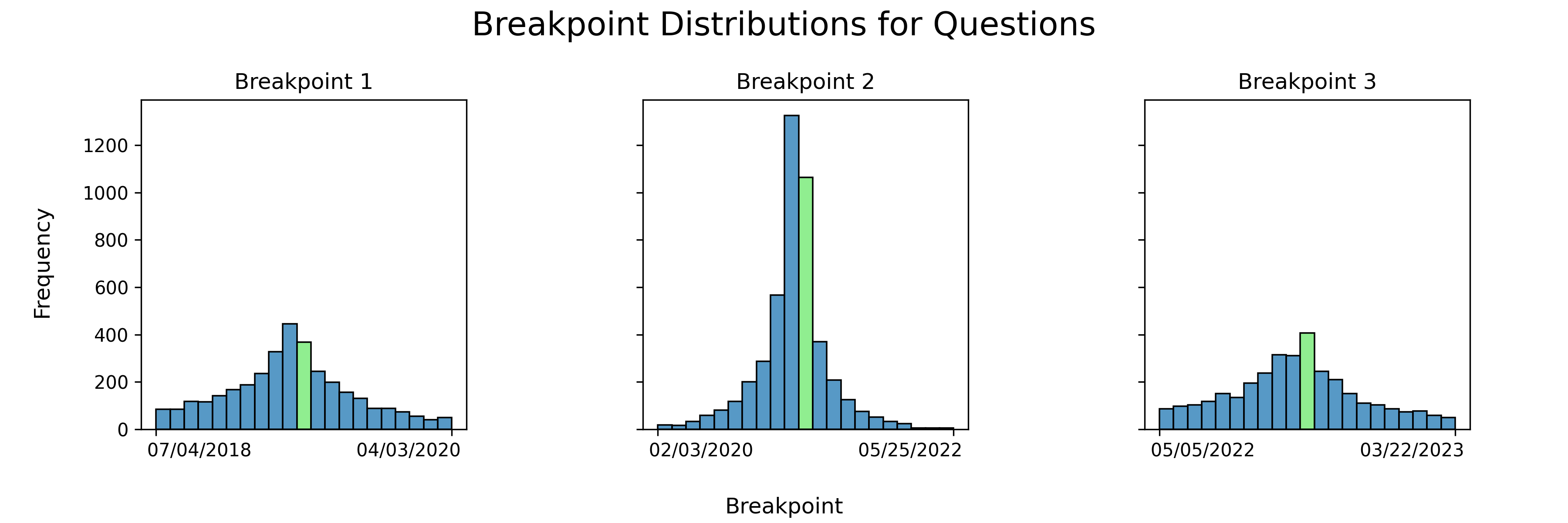}
    \includegraphics[width=0.8\linewidth]{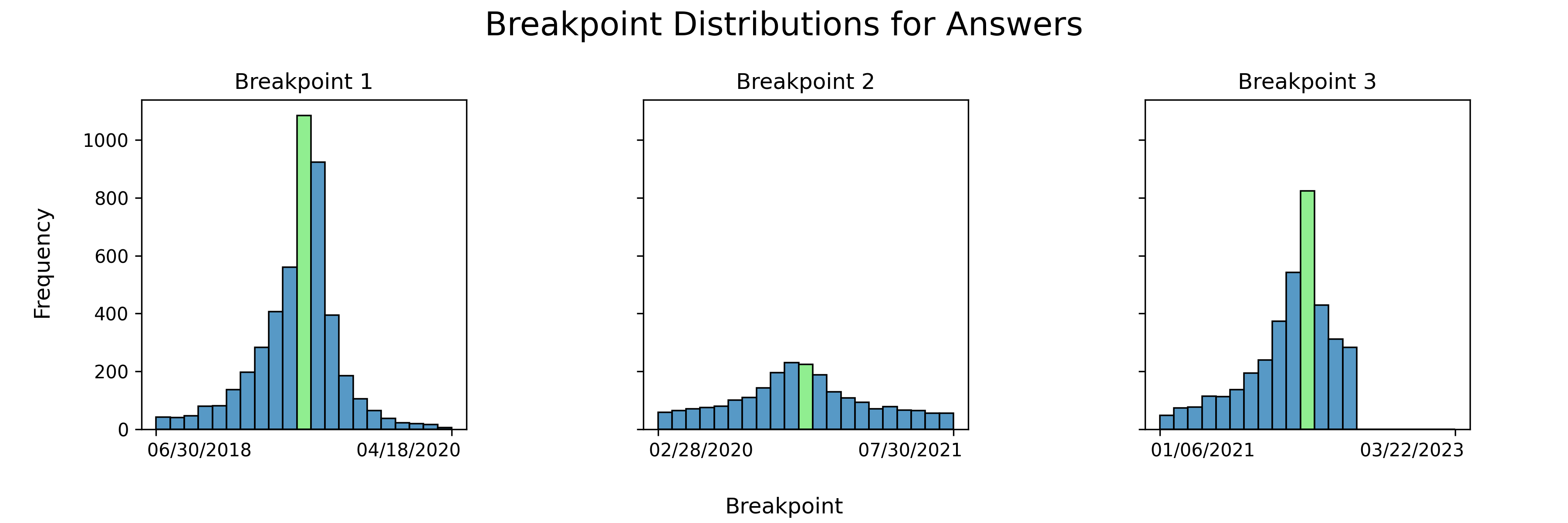}
    \includegraphics[width=0.8\linewidth]{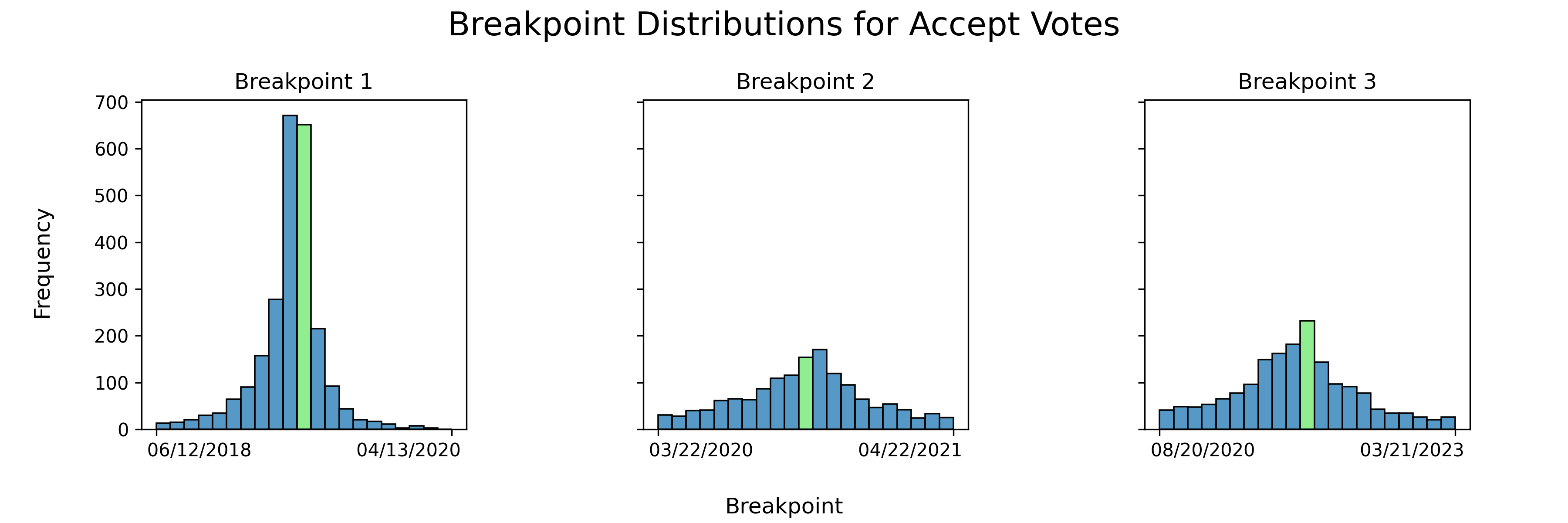}
    \caption{Stack Exchange: Bootstrapping results for Questions (top), Answers (middle), and Accept Votes (bottom) breakpoint estimate highlighted in lightgreen.}
    \label{fig:Bootstrapping}
\end{figure}

\begin{figure}[htbp]
    \centering
    \includegraphics[width=0.48\linewidth]{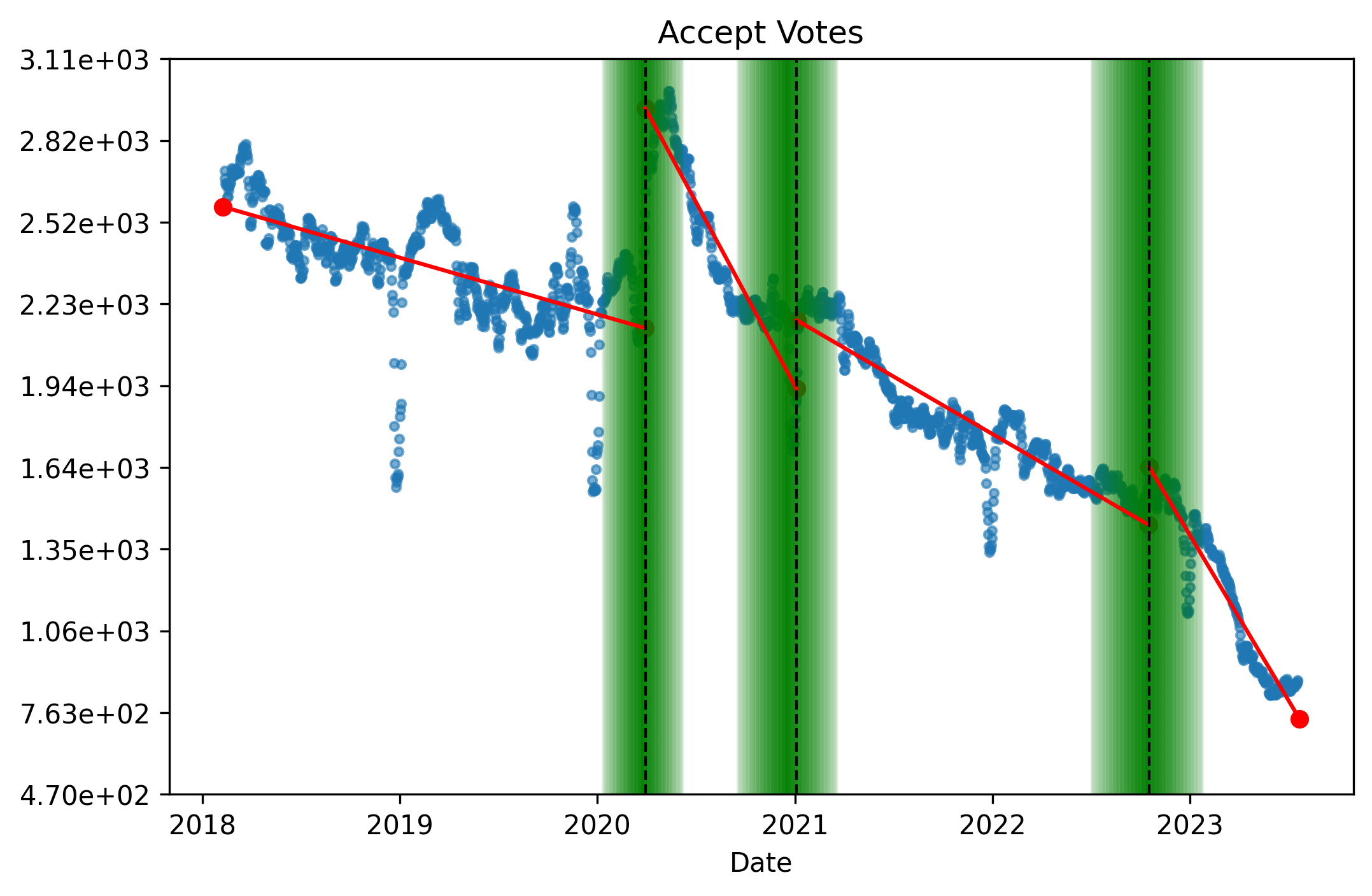}
    \hfill
    \includegraphics[width=0.48\linewidth]{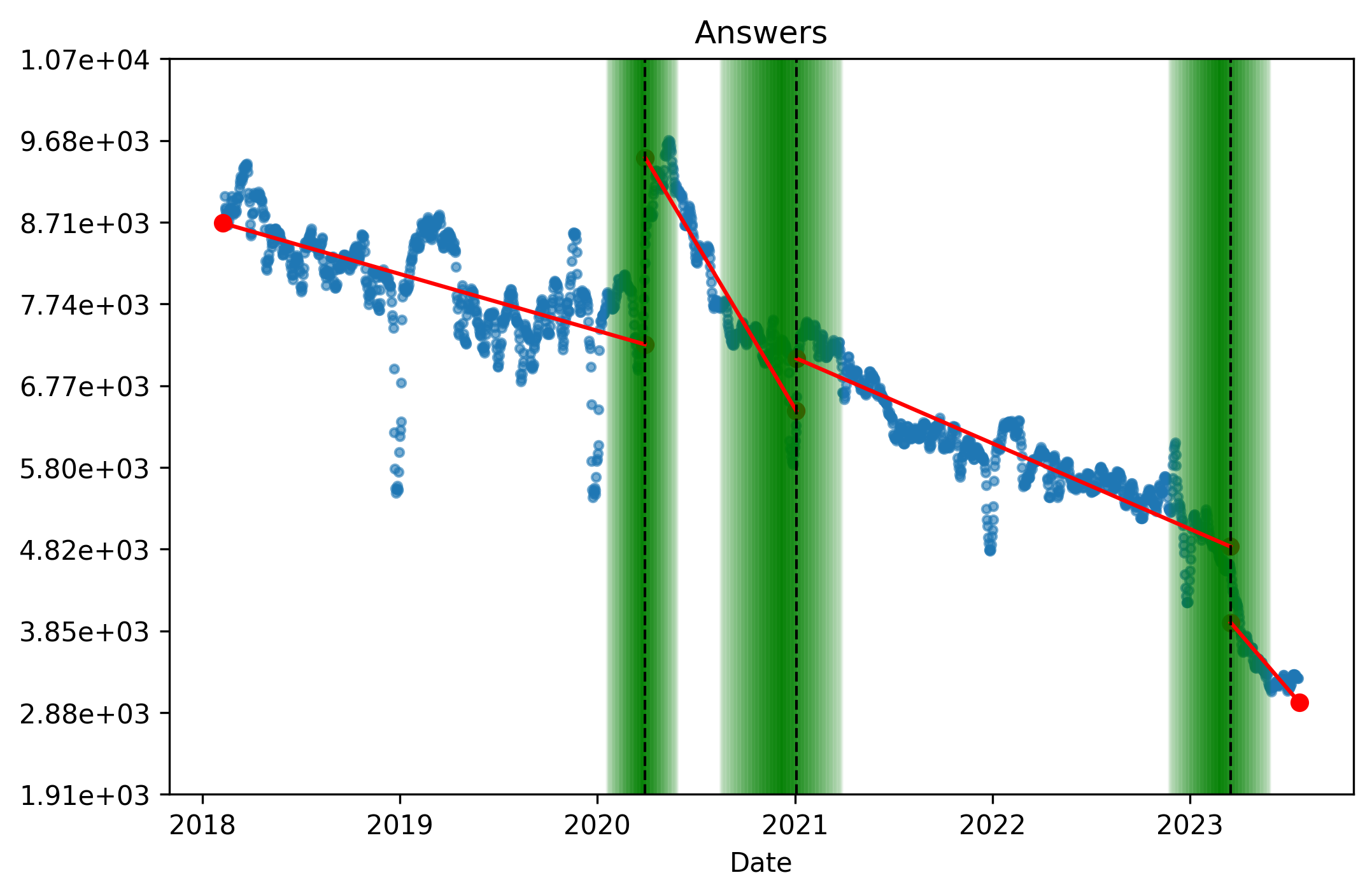}
    \caption{Stack Exchange: Time series plots for Accept Votes (left) and Answers (right) with estimated mean function in red, detected breakpoints as dotted black line, and 95\% confidence intervals in the green shaded region.}
    \label{fig:QA_trends}
\end{figure}

The segmentation results indicate several clear shifts in user activity over time. For questions, four regimes emerge: February 2018-December 2019, December 2019-May 2020, May 2020-October 2022, and October 2022-July 2023. The change at the end of 2019 is followed shortly by the global disruption associated with the COVID-19 pandemic, and the additional transition in May 2020 coincides with the period when remote work and online learning became widespread. These changes likely influenced both the volume and nature of mathematical engagement on the platform. A further structural break in October 2022 appears immediately prior to the public release of ChatGPT (November 2022), suggesting that the emergence of large language models may have altered user behavior in a sustained way.

A similar pattern is visible in the answers and accepted-answer vote series. Both show a pronounced shift in late March 2020, consistent with the early pandemic period, followed by a relatively stable regime through much of 2021 and 2022. All three measures—questions, answers, and accepted votes—identify a structural transition in October 2022, pointing to a broad change in participation rather than an isolated fluctuation in a single metric. Taken together, these results suggest that major external events, particularly the pandemic and the introduction of generative AI tools, coincide with measurable changes in activity on the platform. See Table \ref{Table:RealData} for parameter estimates and confidence intervals.

\begin{table}[!htpb]
\centering
\scriptsize
\caption{Segment-wise regression estimates for Questions, Answers, and Accepted Votes. All values are reported $/100$. For each segment, the table reports slope and intercept estimates with bootstrap standard errors (SE) and 95\% bootstrap percentile confidence intervals.}
\begin{tabular}{ll|cccc}
\multicolumn{2}{r|}{Segment} & 1 & 2 & 3 & 4\\ \midrule
\multicolumn{2}{l|}{Questions} &&&&\\
Slope     & Est. & -1.56 & 25.67 & -3.00 & -10.97 \\
          & SE   & 0.28 & 14.66 & 13.93 & 14.47 \\
          & CI   & (-2.20, -1.15) & (-16.62, 40.86) & (-20.21, 34.30) & (-25.89, 30.84) \\
Intercept & Est. & 8.23 & -10.92 & 11.08 & 26.39 \\
          & SE   & 0.04 & 10.44 & 10.37 & 13.56 \\
          & CI   & (8.17, 8.34) & (-21.92, 19.03) & (-17.60, 23.03) & (-17.99, 35.14) \\
\hline
\multicolumn{2}{l|}{Answers} &&&&\\
Slope     & Est. & -1.85 & -10.21 & -2.81 & -8.29 \\
          & SE   & 0.13 & 4.44 & 3.90 & 3.88 \\
          & CI   & (-2.10, -1.59) & (-14.48, 2.94) & (-12.40, 2.89) & (-13.38, 1.82) \\
Intercept & Est. & 8.70 & 17.35 & 10.07 & 19.46 \\
          & SE   & 0.02 & 4.48 & 3.90 & 5.23 \\\
          & CI   & (8.66, 8.74) & (4.04, 21.60) & (4.21, 19.51) & (3.78, 24.30) \\
\hline
\multicolumn{2}{l|}{Accept Votes} &&&&\\
Slope     & Est. & -0.56 & -3.94 & -1.04 & -3.17 \\
          & SE   & 0.04 & 1.50 & 1.30 & 1.29 \\
          & CI   & (-0.64, -0.47) & (-4.92, 0.96) & (-4.23, 0.85) & (-4.59, 0.48) \\
Intercept & Est. & 2.58 & 6.01 & 3.24 & 7.05 \\
          & SE   & 0.01 & 1.55 & 1.32 & 1.81 \\
          & CI   & (2.56, 2.59) & (1.04, 7.11) & (1.24, 6.40) & (1.08, 8.16) \\
\end{tabular}
\label{Table:RealData}
\end{table}

The segment-wise regression estimates clarify how activity evolved within each detected regime. For Questions, the first segment (February 2018–December 2019) exhibits a modest declining trend, followed by a sharp increase during the early pandemic period (December 2019–May 2020), as reflected by the large positive slope in Segment 2. Activity then transitions to a sustained downward trend through October 2022, with an even steeper decline in the final segment after October 2022. 

The Answers and Accepted-Vote series display similar patterns. Both show a pronounced negative slope beginning in early 2020, consistent with the structural break identified during the onset of the COVID-19 period. In the post\-October 2022 regime, all three measures exhibit significantly negative slopes, indicating a sustained decline in engagement following the emergence of large language model tools. The narrow bootstrap confidence intervals across segments suggest that these trend estimates are statistically stable and not driven by sampling variability.

\section{Discussion and Future Work}

This paper proposed VEGA, a Voronoi-augmented genetic algorithm for high-dimensional derivative-free optimization. The main idea is to combine elitist genetic search with coordinate-wise Voronoi-informed local exploration. By augmenting each elite solution with candidates sampled from its associated local region, VEGA preserves exploitation of high-fitness points while maintaining geometric diversity in the search space. Theoretical results show that, under high-dimensional regularity conditions and a fitness ordering consistent with distance to the target maximizer, the Voronoi mechanism improves the expected convergence behavior relative to standard genetic-algorithm variants. Simulation studies in segmented regression and the Stack Exchange data analysis demonstrate that this improvement translates into more stable breakpoint recovery and useful structural-change detection in practice.

Several theoretical extensions remain important. The current analysis is primarily asymptotic in dimension and is designed to explain why Voronoi-based elite augmentation is useful when high-dimensional random search suffers from poor coverage. A natural next step is to develop fixed-dimensional and non-asymptotic guarantees. Such results would bound the probability that VEGA fails to enter an epsilon-neighborhood of the global maximizer after a finite number of generations, as a function of population size, mutation radius, elite ratio, and the geometry of the objective function. A possible direction is to combine covering-number arguments with finite-sample concentration inequalities for random partitions. This would complement the high-dimensional theory in this paper and provide more explicit guidance for tuning VEGA under a fixed computational budget.

A further limitation concerns pathological or ill-conditioned objective functions. The theoretical justification of VEGA assumes that the elite region continues to contain, or at least remains sufficiently close to, the true maximizer. In highly multimodal, noisy, flat, or misspecified objective functions, this assumption may fail that early elite points may concentrate around a local optimum, and the Voronoi-augmented search may then refine the wrong region. In such cases, elitism can accelerate premature convergence rather than improve global search. Future work should therefore study robust variants of VEGA that preserve a nonzero probability of global exploration. Possible safeguards include multi-start initialization, island-based evolution, adaptive elite-region expansion, occasional uniform resampling, non-vanishing mutation probabilities, and restart rules triggered by stagnation in fitness. A more general theory should replace the assumption that the elite region contains the true maximizer with probabilistic basin-of-attraction conditions, thereby characterizing when VEGA converges globally and when it should be interpreted only as a local derivative-free optimizer.

\clearpage
\begin{appendices}
\section{Definitions, Lemmas and Proofs}
This Appendix discussed how VEGA performs in high dimensional scenarios under limited budget. Section \ref{App:def} provides shorthands for later proofs. Section \ref{App:lem} presents proofs of technical lemmas.
%Lemmas \ref{AsympNormalSample} and \ref{MinimumNormal} quantify the asymptotic behavior of the distance between VEGA estimator and the true value. Based on Lemma \ref{lem:layer-cake} and Lemma \ref{lem:tail-volume}, and Lemmas \ref{lem:lower-expectation-integral}--\ref{Theorem:CloseFormBound} gives bounds of the distance moments and order of estimation error. 
The proofs of main Theorems are given in Section \ref{App:thm}.

\subsection{Definitions} \label{App:def}
\begin{definition}
    Let $\K\subseteq [0,1]^p$ be a set with positive $p$-dimensional volume $|\K|=v$, $0<v\le 1$, and $\X$ be uniformly distributed over $K$. Define $ R=\|\X\|_2.$
    
Since $\K\subseteq [0,1]^p$, we always have $0\le R\le \sqrt p.$ We write
$$
F_p(r) = \left| \left\{\X\in[0,1]^p:\|\X\|_2\le r \right\} \right|,
$$
where $0\le r\le \sqrt p$. Thus $F_p(r)$ is the volume of the intersection of the cube $[0,1]^p$ with the ball of radius $r$ centered at the origin. Define $r_-(v)=F_p^{-1}(v)$, and $r_+(v)=F_p^{-1}(1-v)$, where $F_p^{-1}$ denotes the generalized inverse,
$$
F_p^{-1}(a)=\inf\{r:F_p(r)\ge a\}.
$$
\end{definition}

\begin{definition}
    Define three functions as
    $$
m_p(v) = r_-(v) - \frac1v \int_0^{r_-(v)}F_p(t)\,dt,
$$

$$
M_p(v) = r_+(v) + \frac1v \int_{r_+(v)}^{\sqrt p} \bigl(1-F_p(t)\bigr)\,dt, 
$$
and
$$ 
Q_p(v)= r_+(v)^2 + \frac2v \int_{r_+(v)}^{\sqrt p} t\bigl(1-F_p(t)\bigr)\,dt.
$$
\end{definition}

\subsection{Lemmas and their Proofs} \label{App:lem}

\paragraph{Proof of Lemma \ref{AsympNormalSample}}{
\begin{proof}
    By the independence of $(X_i-a_i)^2$ across $i$, i.e. $(X_i-a_i)^2$ $\bot$ $(X_j-a_j)^2$ for $\forall$ $i \neq j$, and $\Var\left((X_i-a_i)^2\right) < \infty$ for $\forall$ $i = 1, 2, \cdots, p$, $R^2 = \|\X-\a\|^2=\sum_{i=1}^p\left(X_i-a_i\right)^2$ is asymptotically normal distributed as in high dimensional scenarios when $n,p \rightarrow \infty$. To derives its parameters, we further study its expectation and variance.

By rewriting $R^2$, we have
$$
R^2 = \sum_{i=1}^p X_i^2+\sum_{i=1}^p a_i^2-2 \sum_{i=1}^p a_i X_i.
$$

Hence, the expectation is 

\begin{equation}
\begin{aligned}
    \E R^2 &=\sum_{i=1}^p \E X_i^2+\sum_{i=1}^p a_i^2-2 \sum_{i=1}^p a_i \E X_i\\
    &=\sum_{i=1}^p Var X_i+\sum_{i=1}^p\left(\E X_i\right)^2+a_i^2-2 a_i \E X_i \\
    & =\sum_{i=1}^p\left[\frac{1}{12}+\left(a_i-\frac{1}{2}\right)^2\right].
\end{aligned}
\label{DistanceExpectation}
\end{equation}

Similarly, the variance is given by

\begin{equation}
\begin{aligned}
    \Var(R^2) &=\sum_{i=1}^p Var\left(X_i-a_i\right)^2\\
    &=\sum_{i=1}^p \left\{\E\left(X_i-a_i\right)^4-\left[\E\left(X_i-a_i\right)^2\right]^2\right\}.\\
\end{aligned}
\label{DistanceVariance}
\end{equation}

Denote the expectation of $R^2$ as $\mu = r_p^2= p^{-1} \sum_{i=1}^p\left[\frac{1}{12}+\left(a_i-\frac{1}{2}\right)^2\right]$, and the variance $\sigma^2 = p^{-1}\sum_{i=1}^p \left\{\E\left(X_i-a_i\right)^4-\left[\E\left(X_i-a_i\right)^2\right]^2\right\}$, implying $\E R^2 = p \mu$ and $\Var(R^2) = p \sigma^2$. Thus, by \eqref{DistanceExpectation} and \eqref{DistanceVariance}, the asymptotic distribution of $R^2 \overset{\cdot}{\sim} \N\left(\mu p, \sigma^2 p\right)$, or equivalently $p^{-1}R^2 \overset{\cdot}{\sim}  \N\left(\mu, p^{-1} \sigma^2\right)$.

Applying the Delta Method with $g(x) = \sqrt{x}$, we get the asymptotic distribution of $R$ as $p^{-1/2}R \overset{\cdot}{\sim} \N\left(\sqrt{\mu}, p^{-1}\sigma_0^2\right)$, where $\sigma_0^2 = \left[g^{\prime}(\mu)\right]^2 \sigma^2 = \frac{\sigma^2}{4\mu}$. Notate $\mu_p = \sqrt{p\mu}$ and $\sigma_p^2 = \sigma_0^2$, thus $R \overset{\cdot}{\sim} \N\left(r_p, \sigma_p^2\right)$.

Naturally denote that $R_{min} = \min_{j}||\X_j - \a|| = \min_{j}R_j$, where $R_j \overset{\cdot}{\sim} \N\left(\mu_p, \sigma_p^2\right)$ independently. It completes the proof.

\end{proof}
}

\begin{lemma}[Minimum of Gaussian Samples Asymptotics]
Let $Z_1,\dots,Z_n \overset{i.i.d.}{\sim} \N(0,1)$, and define
$$
    L_n = \min_{1\le i\le n} Z_i .
$$
As $n\to\infty$,
$$
    \E[L_n]
    =
    -\sqrt{2\log n}
    +
    \frac{\log\log n+\log(4\pi)-2\gamma}
    {2\sqrt{2\log n}}
    +
    \lowero\!\left(\frac{1}{\sqrt{\log n}}\right),
$$
where $\gamma$ is the Euler--Mascheroni constant. Moreover,
$$
    \Var(L_n)
    =
    \frac{\pi^2}{12\log n}
    +
    \lowero\!\left(\frac{1}{\log n}\right).
$$
In particular,
$$
    \E[L_n] = -\sqrt{2\log n}+\bigO\!\left(\frac{\log\log n}{\sqrt{\log n}}\right)
    \quad \text{and } \quad
    \Var(L_n)=\bigO\left(\frac{1}{\log n}\right).
$$
\label{MinimumNormal}
\end{lemma}

\begin{proof}
Let $M_n:=\max_{1\le i\le n} Z_i$. By symmetry of normal distributions,
$$
    L_n \overset{d}{=} -M_n.
$$
Let
$$
    b_n
    =
    \sqrt{2\log n}
    -
    \frac{\log\log n+\log(4\pi)}
    {2\sqrt{2\log n}}.
$$
The classical extreme-value asymptotics for Gaussian maxima gives
$$
    b_n(M_n-b_n) \dconverge G,
$$
where $G$ has the standard Gumbel distribution. Hence
$$
    \E[M_n]
    =
    b_n+\frac{\gamma}{b_n}+o\!\left(\frac{1}{b_n}\right),
    \qquad
    \Var(M_n)
    =
    \frac{\pi^2}{6b_n^2}
    +
    \lowero\!\left(\frac{1}{b_n^2}\right).
$$
Since $b_n^2\sim 2\log n$ and $L_n\overset{d}{=}-M_n$, we obtain
$$
    \E[L_n]
    =
    -b_n-\frac{\gamma}{b_n}
    +
    \lowero\!\left(\frac{1}{b_n}\right),
$$
and
$$
    \Var(L_n)
    =
    \Var(M_n)
    =
    \frac{\pi^2}{12\log n}
    +
    \lowero\!\left(\frac{1}{\log n}\right).
$$
Substituting the expansion of $b_n$ gives the stated expectation estimate.
\end{proof}

\begin{lemma}[Layer Cake Representation]\label{lem:layer-cake}
Let $\A\subseteq[0,1]^p$ be measurable with $|\A|=v>0$, and let $\Y\sim\operatorname{Unif}(\A)$. If $S=\|Y\|,$ then
$$
\E S = \frac1v \int_0^{\sqrt p} \left| A\cap\{x:\|X\|>t\} \right|\,dt,
$$
and
$$
\E S^2 = \frac1v \int_0^{\sqrt p} 2t \left| A\cap\{x:\|X\|>t\} \right|\,dt.
$$
\end{lemma}

\begin{proof}
For $\forall$ $\X\in[0,1]^p$,
$$
\|\X\| = \int_0^{\sqrt p} \mathone_{\{t<\|\X\|\}}\,dt. $$
Therefore
$$
\int_{\A} \|\X\|\,dx = \int_{\A} \int_{0}^{\sqrt p} \mathone_{\{t<\|\X\|\}}\,dt\,dx.
$$
By Fubini's theorem,
$$
\int_{\A} \|\X\|\,dx = \int_{0}^{\sqrt p} \left| \A\cap\{\X:\|\X\|>t\} \right|\,dt.
$$
Dividing by $v=|\A|$ gives the formula for $\E S$.

Similarly,
$$
\|\X\|^2 = \int_{0}^{\|\X\|}2t\,dt = \int_{0}^{\sqrt p} 2t\mathone_{\{t<\|\X\|\}}\,dt.
$$
Repeating the same argument yields the stated formula for $\E S^2$.
\end{proof}

\begin{lemma}[Tail Volume Bound]\label{lem:tail-volume}
Let $\A\subseteq[0,1]^p$ be measurable with $|\A|=v$. Then, for $\forall$ $0\le t\le\sqrt p$,
$$
\left| \A\cap\{\X:\|\X\|>t\} \right| \ge \bigl(v-F_p(t)\bigr)_+,
$$
and
$$
\left| \A\cap\{\X:\|\X\|>t\} \right| \le \min\{v,1-F_p(t)\}.
$$
\end{lemma}

\begin{proof}
As we have
$$
\left| \A\cap\{\X:\|\X\|\le t\} \right| \le F_p(t),
$$
so
$$
\left| \A\cap\{\X:\|\X\|>t\} \right| = v-\left|\A\cap\{\X:\|\X\|\le t\} \right| \ge v-F_p(t).
$$
Notice that the left-hand side is nonnegative, implying
$$
\left| \A\cap\{\X:\|\X\|>t\} \right| \ge \bigl(v-F_p(t)\bigr)_+.
$$

For the upper bound, clearly
$$
\left| \A\cap\{\X:\|\X\|>t\} \right| \le v,$$
and
$$
\A\cap\{\X:\|\X\|>t\} \subseteq \{\X\in[0,1]^p:\|\X\|>t\},$$
whose volume is $1-F_p(t)$. Hence
$$
\left| \A\cap\{\X:\|\X\|>t\} \right| \le \min\{v,1-F_p(t)\}.
$$
\end{proof}

\begin{lemma}[Expectation Lower Bound]\label{lem:lower-expectation-integral}
Let $\X\sim\operatorname{Unif}(\K)$, where $|\K|=v$. Then
$$
\E\|\X\|
\ge
m_p(v).
$$
\end{lemma}

\begin{proof}
By Lemmas~\ref{lem:layer-cake} and \ref{lem:tail-volume},
$$
\E\|\X\| \ge \frac1v \int_{0}^{\sqrt p} \bigl(v-F_p(t)\bigr)_+\,dt.
$$
By the definition of $r_-(v)=F_p^{-1}(v)$, the positive part is supported on $[0,r_-(v)]$. Hence
$$
\E\|\X\| \ge \frac1v \int_0^{r_-(v)} \bigl(v-F_p(t)\bigr)\,dt.
$$
Therefore
$$
\E\|\X\| \ge r_-(v) - \frac1v \int_0^{r_-(v)}F_p(t)\,dt = m_p(v).
$$
\end{proof}

\begin{lemma}[Expectation Upper Bound]\label{lem:upper-expectation-integral}
Let $\X\sim\operatorname{Unif}(\K)$, where $|\K|=v$. Then
$$
\E\|\X\| \le M_p(v).
$$
\end{lemma}

\begin{proof}
By Lemmas~\ref{lem:layer-cake} and \ref{lem:tail-volume},
$$
\E\|\X\| \le \frac1v \int_0^{\sqrt p} \min\{v,1-F_p(t)\}\,dt.
$$
By the definition of $r_+(v)=F_p^{-1}(1-v)$, we have $1-F_p(t)\ge v$, for $0\le t\le r_+(v)$, and $1-F_p(t)\le v$ for $t\ge r_+(v)$.
Therefore
$$
\E\|\X\| \le \frac1v \left[ \int_0^{r_+(v)}v\,dt + \int_{r_+(v)}^{\sqrt p} \bigl(1-F_p(t)\bigr)\,dt \right].
$$
Thus
$$
\E\|\X\| \le r_+(v) + \frac1v \int_{r_+(v)}^{\sqrt p} \bigl(1-F_p(t)\bigr)\,dt = M_p(v).
$$
\end{proof}

\begin{lemma}[Second Moment Upper Bound]\label{lem:upper-second-moment-integral}
Let $\X\sim\operatorname{Unif}(\K)$, where $|\K|=v$. Then
$$\E\|\X\|^2 \le Q_p(v).$$
\end{lemma}

\begin{proof}
By Lemmas~\ref{lem:layer-cake} and \ref{lem:tail-volume},
$$
\E\|\X\|^2 \le \frac1v \int_0^{\sqrt p} 2t\min\{v,1-F_p(t)\}\,dt.
$$
Splitting the integral at $r_+(v)$ gives
$$
\E\|\X\|^2 \le \frac1v \left[ \int_0^{r_+(v)}2tv\,dt + \int_{r_+(v)}^{\sqrt p} 2t\bigl(1-F_p(t)\bigr)\,dt\right].
$$
The first integration equals $r_+(v)^2$, so
$$
\E\|\X\|^2 \le r_+(v)^2 + \frac2v \int_{r_+(v)}^{\sqrt p} t\bigl(1-F_p(t)\bigr)\,dt = Q_p(v).
$$
\end{proof}

\begin{lemma}[Expectation Closed-form Lower Bound]\label{lem:closed-lower-expectation}
For $\X\sim\operatorname{Unif}(\K)$,
$$
\E\|\X\| \ge \frac p{p+1} \left( \frac{2^p v}{\kappa_p} \right)^{1/p}.
$$
\end{lemma}

\begin{proof}
By definition, 
$
F_p(t) = \left| \{\X\in[0,1]^p:\|\X\|\le t\} \right| \le \frac{\kappa_p}{2^p}t^p,\forall t\geq 0.
$
Let $c_p=\kappa_p/2^p$. By Lemma~\ref{lem:lower-expectation-integral} and $F_p(t)\le c_pt^p$, we have
\begin{eqnarray*}
\E\|\X\| &\ge& \frac1v \int_0^{\sqrt p} \bigl(v-F_p(t)\bigr)_+\,dt\\
&\geq& \frac1v \int_0^{(v/c_p)^{1/p}} \bigl(v-c_pt^p\bigr)\,dt\\
&=& \frac p{p+1} \left(\frac v{c_p}\right)^{1/p}.
\end{eqnarray*}
Since $c_p=\kappa_p/2^p$, we have
$$
\E\|\X\|
\ge
\frac {2p}{p+1}
\left(
\frac{v}{\kappa_p}
\right)^{1/p},
$$
which completes the proof.
\end{proof}

\begin{lemma}[Simplex Volume Estimate]\label{lem:simplex-estimate}
For $\X = (X_1, X_2, \cdots, X_p) \in[0,1]^p$, define
$$
S(\X)=\sum_{i=1}^p(1-X_i).
$$
Then
$$
\E S(\X)
\ge
\frac p{p+1}(p!v)^{1/p}.
$$
\end{lemma}

\begin{proof}
Let $u_i=1-X_i$ and $\u = (u_1, u_2, \cdots, u_p)$. The mapping of $\X\mapsto \u=\mathbf{1}-\X$ preserves volume. Therefore, $\u=\mathbf{1}-\X$ is uniformly distributed on a measurable subset of $[0,1]^p$ with volume $v$.

Define $G_p(t) = \left| \left\{\u\in[0,1]^p: u_1+\cdots+u_p\le t \right\} \right|.$
The set $\{\u\in[0,1]^p:u_1+\cdots+u_p\le t\}$ is contained in the simplex $\{\u\in\mathbb{R}_+^p:u_1+\cdots+u_p\le t\},$
whose volume is $t^p/p!$. Hence
$$
G_p(t)\le \frac{t^p}{p!}.
$$

Applying the layer-cake representation in Lemma~\ref{lem:lower-expectation-integral} to the random variable
$$
S(\u)=u_1+\cdots+u_p,
$$
obviously we have
$$
\E S(\X) = \E S(\u) \ge \frac1v \int_0^p \bigl(v-G_p(t)\bigr)_+\,dt.
$$
Using $G_p(t)\le t^p/p!$ gives
$$
\E S(\X) \ge \frac1v \int_0^{(p!v)^{1/p}} \left(v-\frac{t^p}{p!}\right)\,dt.
$$
Computing the integral yields
$$
\E S(\X) \ge \frac p{p+1}(p!v)^{1/p}.
$$
\end{proof}

\begin{lemma}[Expectation Closed-form Upper Bound]\label{lem:closed-upper-expectation}
For $\X = (X_1, X_2, \cdots, X_p) \sim\operatorname{Unif}(\K)$,
$$
\E\|\X\| \le \sqrt p - \frac{p}{2(p+1)\sqrt p} (p!v)^{1/p}.
$$
\end{lemma}

\begin{proof}
For $\X\in[0,1]^p$,
$$
\sqrt p-\|\X\| = \frac{p-\|X\|^2}{\sqrt p+\|X\|},
$$
because 
$$
p-\|\X\|^2 = \sum_{i=1}^p(1-X_i^2) = \sum_{i=1}^p(1-X_i)(1+X_i) \ge \sum_{i=1}^p(1- X_i) = S(X),
$$
Note that 
$$
\sqrt p+\|\X\|\le 2\sqrt p,
$$
which implies that
$$
\sqrt p-\|\X\| \ge \frac{S(\X)}{2\sqrt p} \quad \text{ and }
\|\X\| \le \sqrt p- \frac{S(x)}{2\sqrt p}.
$$
Taking expectations on both sides gives
$$
\E\|\X\| \le \sqrt p- \frac1{2\sqrt p}\E S(\X).
$$
Therefore
$$
\E\|\X\| \le \sqrt p - \frac{p}{2(p+1)\sqrt p} (p!v)^{1/p},
$$
by Lemma~\ref{lem:simplex-estimate} as $
\E S(\X) \ge p(p!v)^{1/p}/(p+1).$ It completes the proof.
\end{proof}

\begin{lemma}[Second Moment Closed-form Upper Bound]\label{lem:closed-upper-second-moment}
For $\X = (X_1, X_2, \cdots, X_p) \sim\operatorname{Unif}(\K)$,
$$
\E\|\X\|^2
\le
p-
\frac p{p+1}(p!v)^{1/p}.
$$
\end{lemma}

\begin{proof}
For $\forall$ $\X\in[0,1]^p$,
$$
\|\X\|^2 = \sum_{i=1}^p X_i^2 = p- \sum_{i=1}^p(1-X_i^2).
$$
Since
$$
1-X_i^2 =(1-X_i)(1+X_i) \ge 1-X_i,
$$
we have
$$
\|\X\|^2 \le p- \sum_{i=1}^p(1-X_i)= p-S(\X).
$$
Taking expectation on both sides yields
$$
\E\|\X\|^2 \le p- \E S(\X).
$$
Using Lemma~\ref{lem:simplex-estimate},
$$
\E S(\X) \ge \frac p{p+1}(p!v)^{1/p}.
$$
Thus
$$
\E\|\X\|^2 \le p- \frac p{p+1}(p!v)^{1/p},
$$
which completes the proof.
\end{proof}

\begin{lemma}\label{lem:asymptotics-kappa-factorial}
As $p\to\infty$, $
\kappa_p^{1/p}\asymp p^{-1/2}\text{ and }
(p!)^{1/p}\asymp p.
$
.
\end{lemma}

\begin{proof}
By Stirling's formula,
$$
\Gamma\left(\frac p2+1\right)
\asymp
\left(\frac p2\right)^{p/2}e^{-p/2}\sqrt p.
$$
Hence
$$
\kappa_p
=
\frac{\pi^{p/2}}{\Gamma(p/2+1)}
\Rightarrow \kappa_p^{1/p}\asymp p^{-1/2}.
$$
Similarly, we have
$$
p!\asymp \sqrt p\left(\frac pe\right)^p
\Rightarrow
(p!)^{1/p}\asymp p.
$$
\end{proof}




\begin{lemma}
Let $\K\subseteq[0,1]^p$ with $|\K|=v$, where $0<v\le 1$, and let $\X\sim\operatorname{Unif}(\K)$. Then
$$
m_p(v) \le \E\|\X\| \le M_p(v).
$$
Moreover,
$$
\Var(\|\X\|) \le Q_p(v)-m_p(v)^2.
$$
Consequently,
$$
\Var(\|\X\|) \le \min\left\{ \frac p4, \,Q_p(v)-m_p(v)^2
\right\}.
$$
\label{Theorem:IntegralBound}
\end{lemma}

\begin{proof}
The lower bound and upper bound of $\E\|\X\|$ is given by Lemma~\ref{lem:lower-expectation-integral} and Lemma~\ref{lem:upper-expectation-integral} directly.

For the variance, we use the second order moment of
$$
\Var(\|\X\|) = \E\|\X\|^2 - \bigl(\E\|\X\|\bigr)^2.
$$
By Lemma~\ref{lem:upper-second-moment-integral},
$$
\E\|\X\|^2\le Q_p(v),
$$
and by Lemma~\ref{lem:lower-expectation-integral},
$$
\E\|\X\|\ge m_p(v).
$$
Hence
$$
\Var(\|\X\|) \le Q_p(v)-m_p(v)^2.
$$

Since $0\le \|\X\|\le \sqrt p$, Popoviciu's variance inequality shows
$$
\Var(\|\X\|) \le \frac{(\sqrt p-0)^2}{4} = \frac p4.
$$
Combining the two bounds,
$$
\Var(\|\X\|) \le \min\left\{ \frac p4, \,Q_p(v)-m_p(v)^2 \right\}.
$$
\end{proof}

\begin{lemma}
\label{Theorem:CloseFormBound}
Let $\K\subseteq[0,1]^p$ be convex with $|\K|=v$, where $0<v\le1$, and let $\X\sim\operatorname{Unif}(\K)$. Then
$$
\frac p{p+1} \left( \frac{2^p v}{\kappa_p} \right)^{1/p} \le \E\|\X\| \le \sqrt p - \frac{p}{2(p+1)\sqrt p} (p!v)^{1/p}.
$$
Moreover,
$$
\Var(\|\X\|) \le p - \frac p{p+1}(p!v)^{1/p} - \left[ \frac p{p+1} \left( \frac{2^p v}{\kappa_p} \right)^{1/p} \right]^2.
$$
Consequently,
$$
\Var(\|\X\|) \le \min\left\{ \frac p4, \, p - \frac p{p+1}(p!v)^{1/p} - \left[ \frac p{p+1} \left( \frac{2^p v}{\kappa_p} \right)^{1/p} \right]^2\right\},
$$
where $
\kappa_p = \pi^{p/2}/\Gamma(p/2+1)$
denotes the volume of the unit ball in $\mathbb{R}^p$.
\end{lemma}

\begin{proof}
The lower bound for the expectation follows from Lemma~\ref{lem:closed-lower-expectation}, and the upper bound for the expectation follows from Lemma~\ref{lem:closed-upper-expectation}.

For the variance, we again use
$$
\Var(\|\X\|) = \E\|\X\|^2 - \bigl(\E\|\X\|\bigr)^2.
$$
By Lemma~\ref{lem:closed-upper-second-moment},
$$
\E\|\X\|^2 \le p- \frac p{p+1}(p!v)^{1/p}.
$$
By Lemma~\ref{lem:closed-lower-expectation},
$$
\E\|\X\| \ge \frac p{p+1} \left( \frac{2^p v}{\kappa_p} \right)^{1/p}. 
$$
Therefore
$$
\Var(\|\X\|) \le p - \frac p{p+1}(p!v)^{1/p} - \left[ \frac p{p+1} \left( \frac{2^p v}{\kappa_p}\right)^{1/p} \right]^2.
$$
Combining this bound of Popoviciu's inequality
$$
\Var(\|\X\|)\le \frac p4
$$
gives the stated minimum bound.
\end{proof}

\subsection{Proofs of Theorems} \label{App:thm}
\subsubsection*{Proof of Theorem \ref{MinimumDistance}}
\begin{proof}
By Lemma \ref{AsympNormalSample}, $R_i \overset{i.i.d}{\sim} \N\left(\mu_p, \sigma_p^2\right)$ regardless of sample size $n$. By standardization, $\sigma_p^{-1}(R_i - \mu_p) \overset{i.i.d}{\sim} \N\left(0, 1\right)$. Let $R_{min} = min\{R_1, R_2, \cdots, R_n\}$, and it is obvious that $R_{min} = \mu_p+\sigma_p \E Z_{(1)}$, where $Z_{(1)}$ is the minimum normal random variable, i.e. $Z_{(1)} = \min(Z_1, Z_2, \cdots, Z_n)$ and $Z_i \overset{i.i.d}{\sim} \N(0, 1)$.

Take expectation and apply Lemma \ref{MinimumNormal}, we have

\begin{equation*}
    \begin{aligned}
        \E R_{min} & =\mu_p+\sqrt{\sigma_p^2} \E Z_{(1)} \\
& \approx \mu_p+\sigma_p \Phi^{-1}\left(\frac{1}{n+1}\right)\\
&=\mu_p+\sigma_p\left(-\sqrt{2\log n}+\bigO\left(\frac{\log \log n}{\sqrt{\log n}}\right)\right) \\
& =\mu_p-\sigma_p \sqrt{2\log n}+\bigO\left(\frac{\log \log n}{\sqrt{\log n}}\right),
    \end{aligned}
\end{equation*}

where $\Phi(\cdot)$ is the cdf of standard normal distribution. Hence $R_{min} = \mu_p+\sigma_p Z_{(1)}$ implies $R_{min} = \mu_p-\sigma_p \sqrt{2\log n}+\bigO\left(\frac{\log \log n}{\sqrt{\log n}}\right) + \bigO_p(log^{-1/2}n)$ since, by Lemma \ref{MinimumNormal}, $\Var Z_{(1)} = \bigO(\log^{-1}n)$. 
\end{proof}

\subsubsection{Proof of Theorem \ref{MinimumDistanceOrderEst}}

\begin{proof}
    Similarly to the procedure in Lemma \ref{AsympNormalSample},
    
    \begin{equation*}
    \begin{aligned}
        \E R^2 & =\sum_{i=1}^p Var X_i+\sum_{i=1}^p\left(a_i-\E X_i\right)^2 \\
        &=\sum_{i=1}^p\left(a_i-\E X_i\right)^2+\bigO\left(p S^2\right)\\
        &=\|\E \X - \a\|^2 + \bigO\left(p S^2\right),
    \end{aligned}
\end{equation*}
and

\begin{equation*}
        \Var(R^2)  =\sum_{i=1}^p \left\{\E\left(X_i^{(2)}-a_i\right)^4-\left[\E\left(X_i^{(2)}-a_i\right)^2\right]^2\right\} =\bigO\left(p B+p S^4\right),
\end{equation*}

implying $\E R = \|\E \X - \a\|+\bigO\left(p^{-1/2} S\right)$, and $\Var R = \bigO\left(B^{1/2}+S^2\right)$. Further $R^2 = \E R^2 + \bigO_p\left(B^{1/4}+S\right)$.

By Theorem \ref{MinimumDistance} and notations in Lemma \ref{AsympNormalSample},
\begin{equation*}
    \begin{aligned}
R_{min} =& r_p+\sigma_p \E Z_{(1)}\\
=&\|\E \X - \a\|+\bigO\left(p^{\frac{1}{2}} S\right)-\bigO\left(p^{\frac{1}{2}} (B+S^4)^{\frac{1}{2}} \log^{\frac{1}{2}} n\right)\\
&+ \bigO_p\left((B+S^4)^{\frac{1}{4}} \log^{-\frac{1}{2}} n\right) \\
=&\|\E \X - \a\|-\bigO\left(S^2p^{\frac{1}{2}}\log^{\frac{1}{2}} n \right)+ \bigO_p\left(S \log^{-\frac{1}{2}} n\right) 
\end{aligned}
\end{equation*}

$p^{-1}\|\E \X\|$, $p^{-1}\|\E \X^2\| = \bigO(S^2)$, and $p^{-1}\|\E \X^4\| = \bigO(B)$ are bounded, so it implies $S$ and $B$ are both at constant level, which means
$$
R_{min} = \|\E \X - \a\|-\bigO\left(p^{\frac{1}{2}} \log^{\frac{1}{2}} n \right)+ \bigO_p\left(\log^{-\frac{1}{2}} n\right) 
$$
\end{proof}

\subsubsection{Proof of Theorem \ref{Theorem:OrderBounds}}

\begin{proof}
By Lemma~\ref{lem:asymptotics-kappa-factorial},
$$
\kappa_p^{1/p}\asymp p^{-1/2}.
$$
Therefore
$$
\left(
\frac{2^pv}{\kappa_p}
\right)^{1/p}
=
2v^{1/p}\kappa_p^{-1/p}
\asymp
\sqrt p\,v^{1/p}.
$$
Using the closed-form lower bound from Lemma~\ref{lem:closed-lower-expectation}, we obtain $\E\|\X\| \gtrsim \sqrt p\,v^{1/p}.$ Again by Lemma~\ref{lem:asymptotics-kappa-factorial}, $(p!)^{1/p}\asymp p$. Hence
$$
\frac{(p!v)^{1/p}}{\sqrt p}
\asymp
\sqrt p\,v^{1/p}.
$$
Using the upper bound from Lemma~\ref{lem:closed-upper-expectation}, we have, for some universal constant $c>0$, $\E\|\X\| \le \sqrt p - c\sqrt p\,v^{1/p}$. For the variance, we have
$$
\Var(\|\X\|) \le p - \frac p{p+1}(p!v)^{1/p} -  \left[ \frac p{p+1} \left( \frac{2^pv}{\kappa_p} \right)^{1/p} \right]^2.
$$
By Lemma~\ref{lem:asymptotics-kappa-factorial}, we have
$$
\frac p{p+1}(p!v)^{1/p} \asymp pv^{1/p},
$$
and
$$
\left[ \frac p{p+1} \left( \frac{2^pv}{\kappa_p} \right)^{1/p}\right]^2 \asymp pv^{2/p}.
$$
Therefore
$$
Var\|\X\| \lesssim \bigO(p) - \bigO\!\left(pv^{1/p}\right) - \bigO\!\left(pv^{2/p}\right).
$$

\end{proof}

\subsubsection{Proof of Theorem \ref{VEGAConsistency}}
\begin{proof}

By definition, $\C_i(\K) = \{\c \in \K \text{ } | \text{ } \forall j \neq i,  \|\c - \X_i\| \leq \|\c - \X_j\|\}.$

Denote its volume as $V_i=\operatorname{vol}_{[rn]}\left(\C_i(\K)\right)$. Because $\a \notin \partial \K$, samples are around a spherical shell \citep{Li2021}, similarly as shown in Lemma \ref{AsympNormalSample}.

Under high dimension, asymptotic distribution of $\frac{[rn] V_i}{|K|}$ only relies on dimensions $p$, so denote $[rn] V_i/|\K| \rightarrow T_p$, where $\E T_p = 1$ and $Var (T_p) \leqslant 6 \left(\frac{3}{4}\right)^{p / 2}$. See \cite{VoronoiCellMeasure} for a detailed proof. Thus we can write the expression as $[rn] V_i/|\K| = T_p+\lowero_p(1)$. Further,

\begin{equation*}
\begin{aligned}
V_i & =\frac{|\K|}{[rn]} T_p+\lowero_p\left(\frac{|\K|}{[rn]}\right) \\
& =\frac{|\K|}{[rn]}+\frac{|\K|}{[rn]}\left(T_p-1\right)+\lowero_p\left(\frac{|\K|}{[rn]}\right) \\
& =\frac{|\K|}{[rn]}+\lowero_p\left(\frac{|\K|}{[rn]}\left(\frac{3}{4}\right)^{p / 4}\right)+\lowero_p\left(\frac{|\K|}{[rn]}\right).
\end{aligned}
\end{equation*}

For any small $\xi$, the maximum of the volume $M$ satisfies
\begin{equation*}
\begin{aligned}
\max_{1\leq i \leq [rn]} V_i & =\frac{|\K|}{[rn]}+\frac{|\K|}{[rn]} \lowero_p\left(\left(\frac{3}{4}\right)^{p / 4} \sqrt{\log [rn]}\right)+\lowero_p\left(\frac{|\K|}{[rn]} \sqrt{\log [rn]}\right) \\
& =\frac{|\K|}{[rn]}+\lowero_p\left(n^{-1-\xi}\right)
\end{aligned}
\end{equation*}

In VEGA, we generate a elitism region after each generation where the augmented children are located. Considering the elite rate $r$, each elite region has an augmentation size of $rn$ with only one elite individual. Hence, we could use the volume upper bound as order bound $v = |\K|/rn+\lowero_p\left(n^{-1-\xi}\right)$ in Theorem \ref{Theorem:OrderBounds}.

As the initial space is a high dimensional cube, which is convex, so elitism region is also convex. What's more by induction and the assumption that fitness $f(\X_1) > f(\X_2)$ almost surely for candidates closer to true value, $\a$ is always in the elitism region which is still a convex set. Notice the fact that the $\a$ is always in one of the Voronoi cells.

In the Voronoi cell which contain the true maximizer $\a$, there are $[rn]$ augmented samples in crossover step. Hence, we further consider this specific Voronoi cell $\K^{(2)}$. Applying the conditions at the first generation, the minimum distance $R^{(2)}_{min}$ between new samples to $\a$ satisfies the format of $R^{{2}}_{min} = \mu_p+\sigma_p Z_{(1)}$. In high dimensional scenarios, we could approximately take the Voronoi cells as hyper-cones, so the order in Theorem \ref{Theorem:OrderBounds} reduces to

$$
\sqrt p\,v^{1/p} \asymp \E\|\X\|,
$$
and
$$
Var\|\X\| \lesssim \bigO(p) - \bigO\!\left(pv^{1/p}\right) - \bigO\!\left(pv^{2/p}\right),
$$
where $v = \bigO(n^{-1}) + \bigO_p(n^{-1-\xi})$. Hence we have
\begin{equation*}
    \begin{aligned}
        R_{min}^{(2)} \lesssim & \bigO(p^{\frac{1}{2}}n^{-\frac{(1+\xi)}{p}} -p^{\frac{1}{2}}g(m)\log^{\frac{1}{2}} n)\\
        &+\bigO_p(g^{\frac{1}{2}}(m) \log^{-\frac{1}{2}} n).
    \end{aligned}
\end{equation*}

The final minimum distance, considering the elites, is $R_{min} = \min\{R_{min}^{(2)},R_{min}^{(1)}\}$. Since the samples are generated independently,
\begin{equation*}
    \begin{aligned}
        R_{min} \lesssim & \bigO(p^{\frac{1}{2}}n^{-\frac{(1+\xi)}{p}}-p^{\frac{1}{2}}g(m)\log^{\frac{1}{2}} (n + rn)\\
        &+\bigO_p(g^{\frac{1}{2}}(m) \log^{-\frac{1}{2}} (n + rn)).
    \end{aligned}
\end{equation*}

As generation goes to $m$, it further becomes, by induction,
\begin{equation*}
    \begin{aligned}
        R_{min}^{(m)} \lesssim & \bigO(p^{\frac{1}{2}}n^{-\frac{(1+\xi)}{p}}-p^{\frac{1}{2}}g(m)\log^{\frac{1}{2}} (n + mrn)\\
        &+\bigO_p(g^{\frac{1}{2}}(m) \log^{-\frac{1}{2}} (n + mrn)).
    \end{aligned}
\end{equation*}

\end{proof}

\clearpage
\section{Pseudo Codes}

\begin{algorithm}[htbp]
\caption{\textbf{GA Framework}}
\begin{algorithmic}[1]
\Procedure{GeneticAlgorithm}{$\texttt{island}$}
    \State \textbf{Set} Population size, Island number, Elite ratio and tolerance = $tol$
    \State \textbf{Set} max\_iter = 100, mutation rate for Controls
    \State \Call{Initialize}{}
    \For{$g = 1$ to $\texttt{max\_iter}$}
        \If{\texttt{island} = \textbf{false}}
            \State $\texttt{best\_before} \gets$ best fitness in \texttt{generation}
            \State $\texttt{generation} \gets$ \Call{Evolve}{\texttt{generation}}
            \State $\texttt{best\_after} \gets$ best fitness in \texttt{generation}
        \Else
            \State $\texttt{best\_before} \gets \min$ best fitness across islands
            \If{$g \bmod \texttt{island\_crossover} = 0$}
                \State $\texttt{pool} \gets$ flatten all islands into one population
                \State $\texttt{pool} \gets$ \Call{Evolve}{\texttt{pool}}
                \State Shuffle \texttt{pool} and redistribute evenly across islands
            \Else
                \State Apply \Call{Evolve}{} independently within each island
            \EndIf
            \State $\texttt{best\_after} \gets \min$ best fitness across islands
        \EndIf
        \State $\Delta \gets (\texttt{best\_before} - \texttt{best\_after}) / \texttt{best\_before}$
        \If{$\Delta < tol$}
            \State \textbf{break}
        \EndIf
    \EndFor
\EndProcedure
\end{algorithmic}
\label{AlgorithmEvolution}
\end{algorithm}

\begin{algorithm}[htbp]
\caption{Control 0: Evolution Operator}
\begin{algorithmic}[1]
\Procedure{Evolve}{\texttt{batch}}
    \State $m \gets |\texttt{batch}|$
    \State Initialize empty population $\texttt{batch}'$
    \For{$i = 1$ to $m$}
        \State Select two parents uniformly at random from \texttt{batch}
        \State Generate offspring via single-point crossover
        \State Apply additive Gaussian mutation to each gene with fixed probability
        \State Clip offspring to $[0,1]$
        \State Append offspring to $\texttt{batch}'$
    \EndFor
    \State \Return \texttt{batch}' sorted by fitness
\EndProcedure
\end{algorithmic}
\label{Code:Control0}
\end{algorithm}

\begin{algorithm}[htbp]
\caption{Control 1: Evolution Operator}
\begin{algorithmic}[1]
\Procedure{Evolve}{\texttt{batch}}
    \State Let $m \gets |\texttt{batch}|$
    \State Let \texttt{elites} $\gets$ top $\lfloor rm \rfloor$ individuals from \texttt{batch}
    \State Compute selection probabilities over \texttt{batch} proportional to reciprocal fitness
    \State Initialize empty list \texttt{children}
    \For{$i = 1$ to $m - |\texttt{elites}|$}
        \State Sample two parents from \texttt{batch} using the fitness\-based probabilities
        \State Generate offspring via single\-point crossover
        \State Apply mutation to each gene with fixed probability (scaled by parent difference with a minimum step size)
        \State Clip offspring to $[0,1]$
        \State Append offspring to \texttt{children}
    \EndFor
    \State \Return (\texttt{elites} + \texttt{children}) sorted by fitness
\EndProcedure
\end{algorithmic}
\label{Code:Control1}
\end{algorithm}

\begin{algorithm}[htbp]
\caption{Voronoi Partition Search}\label{VPS}
\begin{algorithmic}[1]
\Require Fitness function $\ell(\Z, \btheta ; \X)$, Max search area $\bmS$
\State \textbf{Initialize} Population = $n$, MaxIter = $m$
\State \textbf{Generate} $n$ multidimensional uniform random vector $\Z$ on $\bmS$
\For{$i$ {\color{blue} in} 1:$m$}
\For{$j$ {\color{blue} in} 1:$n$}
\State \textbf{Derive} $\widehat{\btheta}_j = \underset{\btheta \in \bTheta}{\argmax} \quad \ell(\Z_j, \btheta ; \X)$ %{\color{red} e.g., Newton's type or Nelder Mead simplex method}
\State \textbf{Encode} $\Z_j$ to Chromosome($j$)
\EndFor
\State \textbf{Sort} Chromosome descendingly by $\ell(\Z_j, \widehat{\btheta}_j ; \X)$
\State \textbf{Keep} r(\%) with top Fitness
\State By each kept $\Z_j$'s, \textbf{Construct} Voronoi Partitioned area as new $\bmS_j$
$$
\bmS_j = \bmS_{1, j} \times \bmS_{2, j} \times \cdots \times \bmS_{p, j}
$$
\For{$k$ {\color{blue} in} 1:$p$}
\State \textbf{Generate} $\frac{1}{r}$ uniform random number within $\bmS_j$
\State \textbf{Encode} as Chromosomes
\EndFor
\State \textbf{Crossover}, \textbf{Mutation} within each partition

\State \textbf{Amalgamate} all Chromosome together and \textbf{Decode}
\State \textbf{Start} next loop with each $\bmS_j$ as $S$ \footnotemark
\EndFor
\State \textbf{Sort} Chromosome descendingly by $\ell(\Z_i, \widehat{\btheta}_i ; \X)$ \footnotemark
\State \textbf{Return} $\Z_i, \widehat{\btheta}_i$ corresponding to best Chromosome

\end{algorithmic}
\end{algorithm}

\footnotetext[1]{The boundaries of $\bmS_j$ could be complicated after several generations, a possible remedy could be replacing the exact diagram by the circumsphere or circumcube \cite{DAS2021311}.}
\footnotetext[2]{Top candidates (elites) could be kept for monotonicity}

\begin{algorithm}[htbp]
\caption{\textbf{VEGA: Voronoi-Augmented Evolution Operator}}
\begin{algorithmic}
\Procedure{Evolve}{\texttt{batch}}
    \State Let $m \gets |\texttt{batch}|$
    \State Let \texttt{elites} $\gets$ top $\lfloor r m \rfloor$ individuals from \texttt{batch}
    \State Construct rectangular Voronoi bounds $[\texttt{low}_i,\texttt{high}_i]$ around each elite
    \State Initialize empty list \texttt{augmented}
    \For{each elite $i$}
        \State Sample a small set of candidate points uniformly from the hyper\-rectangle $[\texttt{low}_i,\texttt{high}_i]$
        \If{any candidate has finite fitness}
            \State Add the first finite\-fitness candidate to \texttt{augmented}
        \Else
            \State Add the elite itself to \texttt{augmented}
        \EndIf
    \EndFor
    \State \texttt{pool} $\gets$ \texttt{elites} $+$ \texttt{augmented}
    \State Compute selection probabilities over \texttt{pool} proportional to reciprocal fitness
    \State Initialize empty list \texttt{children}
    \For{$i = 1$ to $m - |\texttt{elites}|$}
        \State Sample two parents from \texttt{pool} using the fitness-based probabilities
        \State Generate offspring via single-point crossover
        \State Apply mutation to each gene with fixed probability (scaled by parent difference with a minimum step size)
        \State Clip offspring to $[0,1]$
        \State Append offspring to \texttt{children}
    \EndFor
    \State \Return (\texttt{elites} + \texttt{children}) sorted by fitness
\EndProcedure
\end{algorithmic}
\label{Code:VEGA}
\end{algorithm}

\clearpage
\section{Additional Figures}

\begin{figure}[htbp]
    \centering
    \includegraphics[width=1\linewidth]{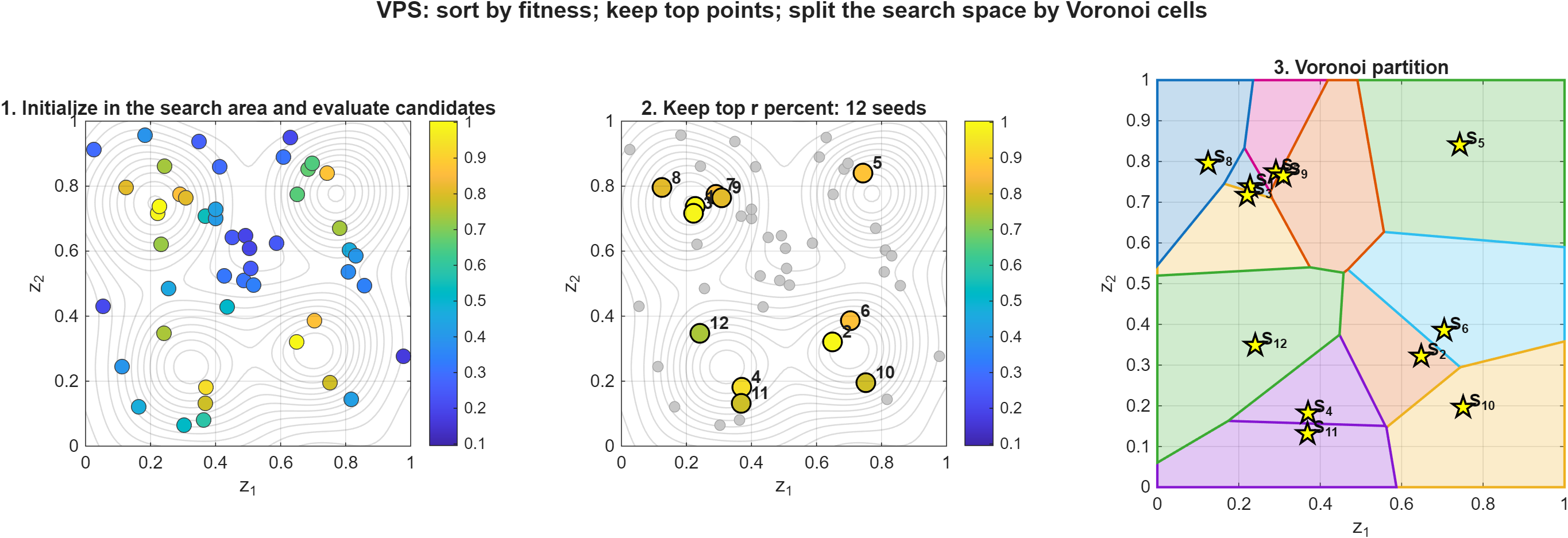}
    \caption{VPS initialization and partitioning: candidates are uniformly sampled and evaluated in the search space, the top r(\%) high-fitness points are retained as seeds, and these seeds define Voronoi cells that partition the space into region-specific search domains.}
    \label{fig:VPS01}
\end{figure}

\begin{figure}[htbp]
    \centering
    \includegraphics[width=1\linewidth]{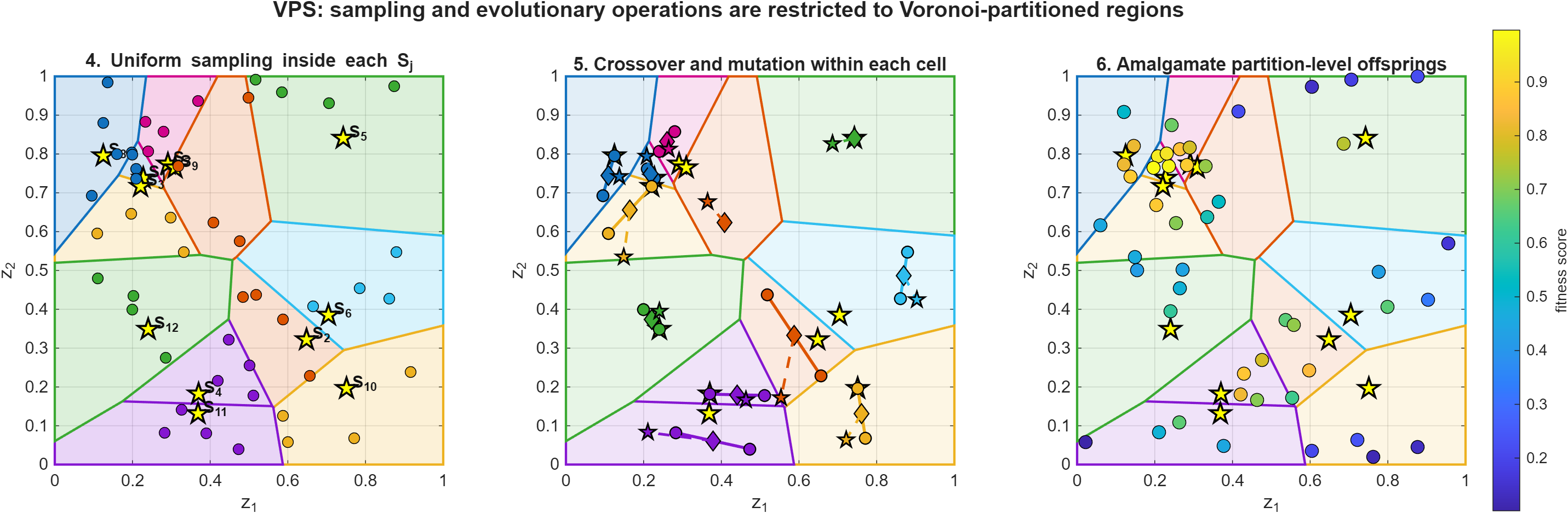}
    \caption{VPS region-restricted evolution: within each Voronoi cell, new candidates are uniformly sampled and evolved by crossover and mutation using only cell-local individuals. Offspring from all cells are then amalgamated into the next population, preserving partition-guided exploration.}
    \label{fig:VPS02}
\end{figure}

\section{Some Potential Applications}

\subsection{Quantile Regression}
Considering $n$ observations with $K$ groups, suppose there're $\n_k$ observations in the $k$-th group, while $\sum_{k = 1}^K n_k \geq n$, so there could be overlappings. For the $i$-th data point, we observe the response $y_i$, a set of the index of groups $s_i = \{k| \text{the } i \text{-th data point is in the } k \text{-th group}\} \subseteq \{1,\ldots,K\}$, i.e., the group membership is known, and covariate vector $\X_i$. We model the quantile regression by minimizing
$$L(\tau, \bbeta) = \sum_{i=1}^n\sum_{k\in s_i} \left[f(\tau_k)  (y_i - \X_i\bbeta_k)^+ + g(\tau_k)  (-y_i + \X_i\bbeta_k)^+\right],$$
where $(x)^+ = max\{x, 0\}$ and (differentiable) functions $f$ and $g$. We denote $\tau = (\tau_1, \tau_2, \cdots, \tau_K)$ and $\bbeta = (\bbeta_1, \bbeta_2, \cdots, \bbeta_k)$ as the quantile and the regression parameters for the $k$-th group respectively..

We call this method quantile regression because it gives an inner perspective of how to model the data quantile by $\tau$. Although we usually given $\tau$ at a level of 0.5, it can be taken as a parameter, i.e., for some specific dataset, 75-th quantile model could be better than a median model for robustness and better prediction. Our goal is to get

$$
(\widehat{\tau}, \widehat{\bbeta}) = \argmin_{\tau, \bbeta} L_\omega(\tau, \bbeta).
$$

The total weighted loss can be easily derived by adding weight to  $L_\omega(\tau, \bbeta) = \sum_{i=1}^n\sum_{k=1}^K \omega_k L_{k, i}$, where $\omega_k$ is given weight of the $k$-th group and $\sum_{k = 1}^K \omega_k = 1$, and solve %\eqref{weightedLoss}

\begin{equation*}
    (\widehat{\tau}, \widehat{\bbeta}) = \argmin_{\tau, \bbeta} L_\omega(\tau, \bbeta)
    %\label{weightedLoss}
\end{equation*}
for the parameter estimate. It is obvious that $\partial L / \partial \tau$ exists with closed form but $\partial L / \partial \bbeta$ doesn't exist due to non-differentiability, so $\tau$ plays as $\btheta$, and $\bbeta$ plays as $\Z$. 

As an illustrative example, let's consider the case when $K = 1$ with $\omega_K = \omega_1 = 1$, i.e., only one group with unknown quantile parameter $\tau_1:=\tau$, i.e. the obejctive function reduces to

\begin{equation*}
L(\tau, \bbeta) = \sum_{i=1}^n\left[f(\tau)  (y_i - \X_i\bbeta)^+ + g(\tau)  (-y_i + \X_i\bbeta)^+\right],
\end{equation*}
for some (differentiable) functions $f$ and $g$. For example, setting $f(\tau)= e^\tau/(1+e^\tau)$ and $g(\tau)=\Phi(\tau)$, where $\Phi(x)$ is the cdf of standard normal distribution. We aim to find a quantile that fits the data best without assuming the underlying noise structure for robustness. Note that we are not going to consider the sub-gradient of $\bbeta$. Moreover, the above problem could also be generalized to the joint quantile regression model \citep{LeeNIPS2016}, which models multiple quartiles simultaneously. Another generalization is modeling the joint conditional quartile of multivariate response \citep{Petrella2019}.

\subsection{Continuous Change Point Detection}

A more challenging problem is change point model with continuous expectation. Here is the model:

$$
Y_{i}(t) =  \X_i(t)\bbeta_{t} + \epsilon_i,
$$
where $\bbeta_{t}=\bbeta_k$ when $\tau_{k}\leq t \leq \tau_{k+1}$, $\epsilon_i \overset{iid}{\sim}N(0,\sigma^2)$ with $-\infty:=\tau_1< \tau_{2}<\cdots<\tau_{K}:=\infty$ and $K$ unknown. Since we have sampling over time, we have to impose additional conditions allowing change point at continuous scale. The nature choice is continuous expectation at change points, i.e.,

$$\X_i(\tau_{k+1})\bbeta_k = \X_i(\tau_{k+1})\bbeta_{k+1},$$
for all $i$ and $k$. Indeed, the parameter estimation given $\tau$'s would still require constrained (numerical) optimization techniques. The choice of $K$ would also required some information criterion, such as the Bayesian information criterion (BIC) and Minimum Description Length (MDL), see \cite{ZhangIC2023} for details about information criteria.

% {\color{red} Jake: Check if the formulation is correct.}

\end{appendices}


\bibliography{References}% common bib file

\end{document}